\def\year{2022}\relax
\documentclass[letterpaper]{article}  
\usepackage{aaai22}  
\usepackage{times}  
\usepackage{helvet}   
\usepackage{courier}   
\usepackage[hyphens]{url}   
\usepackage{graphicx}  
\usepackage{natbib}  
\usepackage{caption} %
\DeclareCaptionStyle{ruled}{labelfont=normalfont,labelsep=colon,strut=off}  
\usepackage{algorithm}
\usepackage{algorithmic}
 
\usepackage{newfloat}
\usepackage{listings}
\floatstyle{ruled}
\newfloat{listing}{tb}{lst}{}
\floatname{listing}{Listing}
 
\nocopyright
 
\newcommand{\var}{\textit{var}\,}
\newcommand{\clause}{\textit{clause}\,}
\newcommand{\tw}{tw}
\newcommand{\apply}{\textup{Apply}}

\newcommand{\calG}{\mathcal{G}}
\newcommand{\calC}{\mathcal{C}}
\newcommand{\calT}{\mathcal{T}}

\newcommand{\calU}{\mathcal{U}}

\newcommand{\Prb}[1]{\textup{Pr}\left[#1\right]}
\newcommand{\Ex}[1]{\textup{E}\left[#1\right]}

\usepackage{caption}
\usepackage{subcaption}

\usepackage{tikz}
\usetikzlibrary{hobby}
\usetikzlibrary{shapes}
\usetikzlibrary{calc}

\usepackage{amsmath}
\usepackage{amssymb}

\usepackage{amsthm}
\newtheorem{theorem}{Theorem}
\newtheorem{lemma}{Lemma}
\newtheorem{corollary}{Corollary}
\newtheorem{claim}{Claim}
\theoremstyle{definition}
\newtheorem{definition}{Definition}
\newtheorem{example}{Example}
 
\title{Lower Bounds on Intermediate Results in Bottom-Up Knowledge Compilation}
\author{
    Alexis de Colnet
    and Stefan Mengel
}
\affiliations{
   	Univ. Artois, CNRS, Centre de Recherche en Informatique de Lens (CRIL), F-62300 Lens, France\\
   	\{decolnet, mengel\}@cril.fr
}

\begin{document}

\maketitle

\begin{abstract}
Bottom-up knowledge compilation is a paradigm for generating representations of functions by iteratively conjoining constraints using a so-called apply function. When the input is not efficiently compilable into a language -- generally a class of circuits -- because optimal compiled representations are \emph{provably} large, the problem is not the compilation algorithm as much as the choice of a language too restrictive for the input. In contrast, in this paper, we look at CNF formulas for which very small circuits exists and look at the efficiency of their bottom-up compilation in one of the most general languages, namely that of structured decomposable negation normal forms (str-DNNF). We prove that, while the inputs have constant size representations as str-DNNF, any bottom-up compilation in the general setting where conjunction and structure modification are allowed takes exponential time and space, since large intermediate results have to be produced. This unconditionally proves that the inefficiency of bottom-up compilation resides in the bottom-up paradigm itself.
\end{abstract}

\section{Introduction}

One of the main objectives of knowledge compilation is transforming, or \emph{compiling}, knowledge given as a CNF formula into other representations, generally subclasses of circuits in decomposable negation normal form (DNNF), which allow for efficient reasoning~\cite{Darwiche01}. There are mainly two approaches to this: \emph{top-down compilation} roughly consists of remembering the trace of an exhaustive backtracking algorithm exploring the whole solution space~\cite{HuangD05}, while \emph{bottom-up compilation} iteratively conjoins representations of the clauses of the input in DNNF. For the latter approach to work, one needs an efficient so-called \emph{apply function} which, given two DNNF and a binary Boolean operation, computes a representation of the function we get by applying the operation on the two DNNF. The only known fragments of DNNF that have such an efficient apply function for conjunctions are so-called structured DNNF (str-DNNF) in which intuitively the variable occurrences in the DNNF must follow a common tree structure called a vtree~\cite{PipatsrisawatD08}. As a consequence, in practice, bottom-up knowledge compilation targets fragments of str-DNNF such as SDD~\cite{Darwiche11,ChoiD13} or OBDD~\cite{Bryant86,somenzi2009cudd}.

One inconvenience of bottom-up compilation that top-down compilation does not have is that it may create intermediate results that are far bigger in size than the final compiled form of the complete input. This was mentioned for OBDD e.g.~in \cite{NarodytskaW07,HuangD04}, and proved for specific bottom-up algorithms compiling  unsatisfiable CNF formulas into OBDD in~\cite{Krajicek08,TveretinaSZ10,FriedmanX13}. As remarked by these works, large intermediate results are problematic because they may lead to failed compilation due to memory outs or very long runtime even for instances that have small representations. The same problem occurs also for the state of the art SDD-compiler of~\cite{ChoiD13}, as can be verified experimentally. To mitigate the problem of large intermediate results, Narodytska and Walsh~(\citeyear{NarodytskaW07}) introduce heuristics for choosing an order in which to conjoin the clauses to try to decrease the size of these intermediate OBDD and show experimentally that these work well when compiling certain configuration problems bottom-up. 

In this paper, we show that having large intermediate results is unavoidable for certain formulas when compiling them bottom-up, even when the final compiled form is of constant size. This is true regardless of the order in which the clauses are conjoined during compilation. We do this by formalizing the bottom-up compilation process into str-DNNF as a deduction process which only uses conjunctive apply and changing of the vtree, also called \emph{restructuring}, a common operation in bottom-up compilation. We then show that in this framework large intermediate results must occur, even when compiling unsatisfiable formulas. 

Informally stated, our main result is the following.

\begin{theorem}[informal]\label{theorem:main_result}
There is a class of CNF formulas that have constant size str-DNNF representations such that any bottom-up compilation must produce intermediate str-DNNF of exponential size.
\end{theorem}

Note that the result of Theorem~\ref{theorem:main_result} is unconditional and does not depend on any unproven complexity assumptions. Moreover, since str-DNNF encompass OBDD and SDD, it is true for bottom-up compilation into these formats.

The formulas that we use to show Theorem~\ref{theorem:main_result} are so-called \emph{Tseitin formulas} which encode certain systems of equations over $\{0,1\}$ whose structure is given by an underlying graph. Tseitin formulas have played a major role in the field of proof complexity, a subfield of theoretical computer science that studies the complexity of refuting unsatisfiable formulas in different proof systems which are often closely linked to practical SAT solvers, see e.g.~\cite{BussN21}. In particular, Tseitin formulas have also been studied when analyzing refutations by proof systems based on different forms of branching programs which are conceptually close to bottom-up compilation, see e.g.~\cite{AtseriasKV04,GlinskihI21,ItsyksonKRS20} for a small sample. Concretely, we here use a recent result from~\cite{deColnetM21} that shows lower bounds on DNNF representations of satisfiable Tseitin formulas. Our basic idea is to show that any bottom-up compilation must essentially construct a DNNF representation of certain sub-formulas of the input that by~\cite{deColnetM21} must be large. In fact, the result of~\cite{deColnetM21} is parameterized by the so-called \emph{treewidth} of the graph underlying the formula, where treewidth is a well-known graph parameter measuring intuitively the treelikeness of a graph. Here, our lower bound is parameterized in the same way, which requires the use of some rather heavy machinery from structural graph theory on the preservation of treewidth under graph partitions.

\section{Preliminaries}

A Boolean variable $x$ is a variable taking its value in $\{0,1\}$. A \emph{literal} is a variable $x$ or its negation $\overline{x}$. An \emph{assignment} to a set of variables $X$ is a mapping from $X$ to $\{0,1\}$. A Boolean function $f$ on $X$ is a mapping of the assignments to $X$ to $\{0,1\}$. The \emph{satisfying assignments} of $f$ are the assignments mapped to $1$ by $f$. Two functions on $X$ are equivalent, written $f \equiv g$, when their satisfying assignments are the same. When $X$ is not specified, $\var(f)$ denotes the set of variables of $f$. Given an assignment $a$ to $Y \subseteq X$, the function $f$ \emph{conditioned on} $a$, written $f|a$, is the function on $X \setminus Y$ obtained from $f$ after fixing all variables in $Y$ to their values given by $a$. As usual, the symbols $\lor$ and $\land$ denote disjunction and conjunction, respectively. A \emph{clause} is a disjunction of literals and a CNF formula (Conjunctive Normal Form) is a conjunction of clauses. The set of clauses of a CNF formula $F$ is denoted by $\clause(F)$. We say that $F'$ is a subformula of $F$ when $\clause(F') \subseteq \clause(F)$. The formula $F'$ is called a \emph{proper} subformula when the inclusion is strict.

\subsection{Structured Decomposable Negation Normal Forms}

A Boolean circuit $\Sigma$ is a directed acyclic computation graph without parallel edges, whose leaves are labeled by literals or Boolean constants $0$ or $1$, and whose internal nodes are labeled by Boolean operations. The size of $\Sigma$, denoted by $|\Sigma|$, is its number of edges. The set of variables whose literals label the leaves under a node $s$ is written $\var(s)$. Each node~$s$ computes a Boolean function on $\var(s)$ defined in the obvious inductive way. The function computed by $\Sigma$ is defined as that computed by its roots. 

A node $s$ with successors $s_1,\dots,s_k$ is called \emph{decomposable} when $\var(s_i) \cap \var(s_j) = \emptyset$ holds for all $i \neq j$. A \emph{Decomposable Negation Normal Form} (short DNNF) for a function $f$ is a Boolean circuit computing $f$, whose internal nodes are labeled with $\lor$ or $\land$ and such that all $\land$-nodes are decomposable. The \emph{DNNF language} is the class of DNNF circuits. One can modify a DNNF in linear time without altering the function it computes so that every internal node has fan-in 2. So we assume that all DNNF in this paper have only internal nodes with fan-in 2.
 
Let $X$ be a finite set of Boolean variables. A \emph{vtree} $T$ for~$X$ is a binary tree whose leaves are in bijection with $X$. For $t \in T$, we denote by $\var(t)$ the set of variables corresponding to the leaves under $t$. A \emph{structured DNNF} (str-DNNF) is~a DNNF $\Sigma$ equipped with a vtree $T$ on its variables and a mapping $\lambda$ from the nodes of $\Sigma$ to that of $T$ such that:
\begin{itemize}
\item[1.] for every $\land$-node $s$ with successors $s_0$ and $s_1$, if $\lambda(s) = t$, then $t$ is an internal node of $T$ and there are $t_l$ and $t_r$ rooted under the two children of $t$ such that $\lambda(s_0) = t_l$ and $\lambda(s_1) = t_r$
\item[2.] for every $\lor$-node $s$ with successors $s_0$ and $s_1$, there is $\lambda(s) = \lambda(s_0) = \lambda(s_1)$
\item[3.] for every $s$, $\var(s) \subseteq \var(\lambda(s))$ holds
\end{itemize}
$\Sigma$ is said to be \emph{structured} by $T$, or to \emph{respect} the vtree $T$. Given any vtree $T$ on variables $X$, all Boolean functions on $X$ are computed by some str-DNNF respecting $T$: just write the function in DNF (Disjunctive Normal Form) and see that every term can be turned into a str-DNNF respecting $T$. We remark that both SDD and OBDD are restricted forms of str-DNNF~\cite{DarwicheM02,Darwiche11}.


Enforcing structuredness for DNNF can in theory result in a size blow up~\cite{PipatsrisawatD10}, but it has some very useful benefits. On the one hand, in certain fragments it allows for canonicity which is often desirable~\cite{BroeckD15}. 
On the other hand, structuredness is the only known property that yields efficient algorithms for conjoining DNNF~\cite{PipatsrisawatD08}: there is an algorithm $\apply(\Sigma_1,\Sigma_2,\land)$ that, given two str-DNNF $\Sigma_1$ and $\Sigma_2$ respecting the same vtree, returns a str-DNNF equivalent to $\Sigma_1 \land \Sigma_2$ with the same vtree as $\Sigma_1$ and $\Sigma_2$, and runs in time $O(|\Sigma_1|\times|\Sigma_2|)$. So, consider a situation in which the clauses of a CNF formula are split into $F_1$ and $F_2$, and assume the str-DNNF $\Sigma_1$ and $\Sigma_2$ compute $F_1$ and $F_2$, respectively, and respect the same vtree. Then finding a str-DNNF that computes $F$ is feasible in quadratic-time as it boils down to running $\apply(\Sigma_1,\Sigma_2,\land)$. This is the key principle behind bottom-up compilation.

\subsection{Bottom-Up Compilation}

Let $L$ be a compilation language like str-DNNF. We formalize a bottom-up compilation of CNF formula $F = C_1 \land \dots \land C_m$ as a sequence of circuits in $L$, $\Sigma_1, \dots, \Sigma_N$, culminating in $\Sigma_N \equiv F$ and such that, for all $i \in [N]$ 
\begin{itemize}
\item[•] $\Sigma_i\equiv C_j$ for some clause $C_j$ in $F$, $j \in [m]$, or
\item[•] $\Sigma_i = \apply(\Sigma_j,\Sigma_k,\land)$ with $j,k < i$ and $\Sigma_j$ and $\Sigma_k$ have the same vtree, or
\item[•] $\Sigma_i \equiv \Sigma_j$ with $j < i$ and the vtrees for $\Sigma_i$ and $\Sigma_j$ differ.
\end{itemize}
Note that $\Sigma_i \equiv \Sigma_j$ is not necessarily easily verifiable in our framework. We say that we have an $L(\land,r)$ compilation of $F$, where~$r$ indicates that vtree modification (\emph{r}estructuring) is allowed. We call an \emph{$L(\land,r)$ refutation} any $L(\land,r)$ compilation of an unsatisfiable formula. In this paper we will focus on $\textup{str-DNNF}(\land,r)$ compilations and refutations.

We are interested in the amount of memory used when compiling CNF formulas bottom-up. To abstract away implementation details, we note that in any case a bottom-up compiler must keep every $\Sigma_i$ in memory at some point\footnote{Note that the whole sequence $\Sigma_1,\dots,\Sigma_N$ never has to be kept in memory entirely since earlier $\Sigma_i$ can be deleted from memory when they are not needed anymore~\cite{BussN21}.}. Thus, the size of the biggest $\Sigma_i$ is a lower bound on the space needed, and thus also on the time taken, by the compilation.
One can then envision a compilation whose final circuit is way smaller than the biggest intermediate circuit, i.e., $|\Sigma_N| \ll \max_{i\in [N]} |\Sigma_i|$. Then, the run of the bottom-up compiler leading to the sequence appears intuitively wasteful. This is most visible when compiling unsatisfiable CNF formulas: the smallest compiled form is a single node $0$, and since satisfiability testing is tractable in $L$, we can assume that $\Sigma_N = 0$, and yet, its bottom-up compilation may have large memory cost.

Note that the size of the $\Sigma_i$ can differ dramatically depending on the sequence of apply operations, i.e., the order in which the clauses are conjoined. However, we will see that there are formulas that have constant-size str-DNNF representation but for which every possible $\textup{str-DNNF}(\land,r)$ compilation must produce big intermediate results.

\subsection{Graphs}

We assume that the reader is familiar with basics and notation from graph theory as e.g.~found in~\cite{Diestel12}. In this section, we will remind the reader of some concepts that will be used in the remainder of this paper.

Graphs in this paper are undirected, do not contain self-loops, but may have parallel edges. Given a graph $G=(V,E)$ and a set $A\subseteq V$, we denote by $G-A$ the graph we get from $G$ by deleting all vertices in $A$ and all edges that contain a vertex in $A$. If $A$ consists of a single node $u$, we also write $G-u$ instead of $G-\{u\}$. By $G[A]$ we denote the graph induced by $A$ in $G$, i.e., the graph $G- (V\setminus A)$. Given another set $B \subseteq V$, we denote by $E(A,B)$ the set of edges of $G$ that have one endpoint in $A$ and the other in $B$.

A graph is called connected if there is a path from every vertex to every other vertex. A connected component is defined as a maximal connected subgraph. A $1$-separator of a connected graph $G$ is defined to be a vertex $u$ such that $G-u$ is not connected. A graph is called $2$-connected if it is connected, has at least two vertices and contains no $1$-separator.

The \emph{treewidth} $\tw(G)$ of a graph $G$ is a well-known graph parameter with broad applicability in artificial intelligence that measures roughly how close $G$ is to being a tree. Since we will not need its technical definition in this paper but use several results on it as black boxes, we will not formally introduce it here and refer the reader to~\cite{Diestel12,HarveyW17}. We will use the following result from~\cite{BodlaenderK06} which we reformulate to simplify notation.

\begin{theorem}\label{theorem:bodlaender}
 Let $G$ be a graph with a $1$-separator $u$. Then $G-u$ contains a connected component $G'= (V', E')$ such that $\tw(G) = \tw(G[V'\cup \{u\}])$.
\end{theorem}

\subsection{Tseitin formulas}


We study Tseitin formulas which are CNF formulas representing systems of parity constraints structured by a graph $G = (V,E)$. The graph is equipped with a function $c : V \rightarrow \{0,1\}$ which assigns \emph{charges} 0 or 1 to its vertices. Each edge $e$ of $G$ is associated to a Boolean variable~$x_e$. Given a set $E' \subseteq E$, we write $X_{E'} = \{x_e \mid e \in E'\}$. The Tseitin formula encodes the fact that, if we only keep in $G$ the edges whose variables are given value 1, then all vertices with charge 1 have an odd degree and all vertices with charge 0 have an even degree. More formally let $E(v)$ denote the set of edges of which $v$ is an endpoint and define the constraint 
$$
\chi_{v,c} : \sum_{e \in E(v)} x_e = c(v) \mod 2
$$
then the Tseitin formula $T(G,c)$ computes $\bigwedge_{v \in V} \chi_{v,c}$. Each $\chi_{v,c}$ can be encoded in a CNF formula $F_{v,c}$ on variables $X_{E(v)}$ composed of $2^{|E(v)|-1} = 2^{\deg(v)-1}$ clauses of size $\deg(v)$. The Tseitin formula over $G$ for the charge function $c$ is the CNF formula $T(G,c) = \bigwedge_{v \in V} F_{v,c}$. For convenience we often drop $c$ from the notations writing only $T(G)$, $\chi_v$, or $F_v$. For $v \in V$ let $1_v : V \rightarrow \{0,1\}$ be the function mapping $v$ to $1$ and all other vertices to $0$. The complement parity constraint to $\chi_{v,c}$ is $\chi_{v,(c+1_v\!\!\mod 2)}$, which we write $\overline{\chi}_{v,c}$ for convenience.

We use the notation $\clause(\chi_{v,c})$ to denote the set of clauses of $F_{v,c}$. We extend this notation to Tseitin formulas by defining $\clause(T(G,c)) = \bigcup_{v \in V} \clause(\chi_{v,c})$.

\begin{example}
Let $G$ be the graph 
\raisebox{-0.5\height}{
\begin{tikzpicture}[scale=0.8, every node/.style={scale=0.9}]
\def\s{2.5};
\node[circle, fill=black, draw=white, inner sep=\s, text=white] at (0,0) (a) {};
\node[circle, fill=white, draw=black, inner sep=0.75*\s, text=black] at (0.9,0.9) (b) {};
\node[circle, fill=black, draw=white, inner sep=\s, text=white] at (1.3,0) (c) {}; 

\draw (a) -- node[near end, left] {$x\,\,$} (b);
\draw (b) -- node[midway, right] {$z$} (c);
\draw (a) to  [out=-10,in=-170] node[midway, below] {$y$} (c);
\end{tikzpicture}
} where white vertices have charge 1 and black vertices have charge 0. The corresponding Tseitin formula is $T(G) = (x \lor z) \land (\overline{x} \lor \overline{z}) \land (x \lor \overline{y}) \land (\overline{x} \lor y) \land (y \lor \overline{z}) \land (\overline{y} \lor z)$.
\end{example}

There is a simple criterion for the satisfiability of Tseitin formulas.
\begin{lemma}[\citealt{Urquhart87}]\label{lemma:Tseitin_formula_satisfiability}$T(G,c)$ is satisfiable if and only if $\sum_{v \in U} c(v) = 0 \mod 2$ holds for all connected components $G' = (U,E')$ of $G$.
\end{lemma}
\noindent In this paper we study the space complexity of $\textup{str-DNNF}(\land,r)$-compilation of unsatisfiable Tseitin formulas whose underlying graph is connected. We parameterize our bounds by the treewidth of the graph. For exponential lower bounds to be relevant, we need an input CNF formula whose length is polynomial in the number of variables. We achieve this by restricting our study to graphs of maximum degree bounded by some constant $\Delta$. This very common restriction leads to an upper bound of $|V|\times 2^{\Delta-1}$ on the number of clauses in $T(G)$.

Note that there is always a small str-DNNF for a single parity constraint.
\begin{lemma}[\citealt{PipatsrisawatD10}] \label{lemma:parity_constraint_to_strDNNF}
Let $\chi$ be a parity constraint and let $T$ be a vtree on $\var(\chi)$. There is a str-DNNF of size $O(|\var(\chi)|)$ respecting $T$ that computes~$\chi$.
\end{lemma} 

\noindent However representing a \emph{satisfiable} Tseitin formula in str-DNNF, so a system of parity constraints, is expensive. 

\begin{theorem}[\citealt{deColnetM21}]\label{theorem:sat_Tseitin_formulas_to_DNNF}
The smallest DNNF representing $T(G)$ \emph{satisfiable} with $G$ a graph of maximum degree $\Delta$ has size at least $2^{\Omega(k)/\Delta}/n$ with $k = \tw(G)$ and $n = |\var(T(G))|$.  
\end{theorem}

\section{Refuting Tseitin formulas in str-DNNF($\land$, $r$)}

In this section, we will give the formal version of our main result Theorem~\ref{theorem:main_result} and prove it, building on several lemmas whose proof we defer to the following sections. We start with a simple observation that essentially says that, given a bottom-up compilation of a function $f$, one can easily infer a bottom-up compilation of $f|a$, for any partial assignment $a$. This will be useful in several upcoming proofs.

\begin{lemma}\label{lemma:conditioning_for_strDNNF_compilation}
Let $F$ be a CNF formula and $\Sigma_1,\dots,\Sigma_N$ be a str-DNNF($\land$, $r$) compilation of $F$. Let $a$ be a partial assignment to $\var(F)$, then $\Sigma_1|a,\dots,\Sigma_N|a$ is a str-DNNF($\land$, $r$) compilation of $F|a$.
\begin{proof}
For every $i$ between $1$ and $N$ let $\Sigma'_i$ be $\Sigma_i|a$. str-DNNF allow linear-time conditioning without size increase nor vtree modification, so $|\Sigma'_i| \leq |\Sigma_i|$ and $\Sigma'_i$ and $\Sigma_i$ share a common vtree. We have $\Sigma_N \equiv F$, so $\Sigma'_N \equiv F|a$ follows. We will prove that, for every $i$, either $\Sigma'_i$ is the str-DNNF representation of a clause of $F|a$, or there are $j,k < i$ such that $\Sigma'_i = \textup{Apply}(\Sigma'_j,\Sigma'_k,\land)$ where all three str-DNNF share a common vtree, or there is $j < i$ such that $\Sigma'_i \equiv \Sigma'_j$ and the vtree of $\Sigma'_i$ and $\Sigma'_j$ may differ.

Take an arbitrary $i$ between $1$ and $N$. If $\Sigma_i$ is the str-DNNF representation of a clause $C \in F$, that is, $\Sigma_i \equiv C$, then clearly $\Sigma'_i = \Sigma_i |a \equiv C|a$ and $C|a$ is indeed a clause of $F|a$. Otherwise if $\Sigma_i$ is the str-DNNF returned by $\textup{Apply}(\Sigma_j,\Sigma_k,\land)$, then $\Sigma_i \equiv \Sigma_j \land \Sigma_k$ and all three str-DNNF share a common vtree. Then $\Sigma'_i = \Sigma_i|a \equiv (\Sigma_j \land \Sigma_k)|a \equiv \Sigma_j|a \land \Sigma_k|a = \Sigma'_j \land \Sigma'_k$. Since the vtree is not modified by conditioning we can feed $\Sigma'_j$ and $\Sigma'_k$ to an Apply to obtain $\Sigma'_i = \textup{Apply}(\Sigma'_j,\Sigma'_k,\land)$. Finally in the case where $\Sigma_i$ is equivalent to $\Sigma_j$ with potentially a vtree modification, it is clear that $\Sigma'_i = \Sigma_i|a \equiv \Sigma_j|a = \Sigma'_j$.
\end{proof}
\end{lemma}

Our main result is the following theorem on the refutation of unsatisfiable formulas by bottom-up compilation.

\setcounter{theorem}{0}
\begin{theorem}
Let $\calG$ be a class of graphs whose maximum degree is bounded by a constant. All str-DNNF($\land$, $r$) refutation of an unsatisfiable $T(G)$ with $G \in \calG$ have size at least $2^{\Omega(k)}poly(1/n)$ with $k = \tw(G)$ and $n = |\var(T(G))|$.
\end{theorem}
\setcounter{theorem}{3}

We will prove Theorem~\ref{theorem:main_result} later in this section after some discussion and preparations.
First, note that there are graphs of bounded degree with treewidth linear in the number of vertices, see e.g.~\cite{GroheM09}. It follows that there are formulas where the intermediate results have exponential size.

\begin{corollary}
There is a family of unsatisfiable CNF formulas such that every formula on $n$ variables has $O(n)$ clauses and all its str-DNNF($\land$, $r$) refutations have an intermediate result of size~$2^{\Omega(n)}$.
\end{corollary}

Let us compare Corollary 1 with known exponential lower bounds on the size of intermediate results for similar refutation systems, see for instance~\cite{Krajicek08,Segerlind08,TveretinaSZ10,FriedmanX13}. First, we are not aware of refutation systems using str-DNNF circuits that are not OBDD or branching programs. Since OBDD are generally exponentially bigger than str-DNNF, our result is stronger in that respect. Moreover, restructuring is rarely allowed in the OBDD-based proof system while it is in ours. Most known bounds are stated for OBDD-based refutations in which the variable order can be arbitrary but cannot be changed in the refutation.
Also we do not require any specific order in which the clauses are conjoined, which is a restruction used for some bounds in, e.g.,~\cite{FriedmanX13}.

Our results might look somewhat unconvincing since they only talk about the compilation of unsatisfiable formulas, a setting in which costly compilation can be substituted by a usually much less expensive single call of a SAT solver\footnote{In fact, some knowledge compilers, e.g. the top-down knowledge compiler D4~\cite{LagniezM17}, make a call to a SAT solver before trying to compile the input to avoid wasting time when compiling unsatisfiable instances.}.
However, equipped with Lemma~\ref{lemma:conditioning_for_strDNNF_compilation}, we can lift them to satisfiable formulas that have constant size str-DNNF representation with a simple trick.

\begin{corollary}\label{corollary:main_sat}
There are satisfiable CNF formulas that have constant size str-DNNF representations such that any str-DNNF($\land$, $r$) compilation must have an intermediate result of size~$2^{\Omega(n)}$ where $n$ is the number of variables in the input.
\begin{proof}
Consider a class of unsatisfiable Tseitin formulas $\calT := \{T(G) \mid G \in \calG\}$ for a class of graphs $\calG$ of treewidth linear in the number of vertices and let $x$ be a fresh variable not used in any of these formulas. For each $T(G)$ let $F(G)$ be the formula $T(G)$ with the additional literal $x$ added to all clauses.  Clearly, $F(G) \equiv x \lor T(G) \equiv x$, so the smallest str-DNNF representing $F(G)$ has size $1$. By Lemma~\ref{lemma:conditioning_for_strDNNF_compilation}, given a str-DNNF($\land$,$r$) compilation of $F(G)$, we can condition all intermediate str-DNNF on $x = 0$ to obtain a str-DNNF($\land$,$r$) refutation of $T(G)$. Since conditioning does not increase the size of str-DNNF the corollary follows from Theorem~\ref{theorem:main_result}.
\end{proof}
\end{corollary}

Note that we could prove a version of Corollary~\ref{corollary:main_sat} parameterized by the so-called primal treewidth of the formulas. Since we do not want to introduce even more notions, we abstain from doing so here.


As a first step towards the proof of Theorem~\ref{theorem:main_result}, let $T(G)$ be unsatisfiable with $G = (V,E)$ connected. We look at the very last Apply in the refutation of $T(G)$ in $\textup{str-DNNF}(\land,r)$: 
$$
\Sigma_N = \apply(\Sigma^\ell,\Sigma^r,\land)
$$
where $\Sigma_N \equiv 0$ and $\Sigma^\ell$ and $\Sigma^r$ are two satisfiable str-DNNF structured by the same vtree. Roughly put, the proof of Theorem~\ref{theorem:main_result} is as follows:
\begin{itemize}
\item[1.] We prove that there is a partition $(A,B)$ of $V$ such that \emph{both} $G[A]$ \emph{and} $G[B]$ have treewidth $\Omega(\tw(G))$.
\item[2.] For that partition we show how to construct from $\Sigma^\ell$ and $\Sigma^r$ in polynomial time a str-DNNF $\Sigma^*$ computing a \emph{satisfiable} Tseitin formula $T(G[A])$ or $T(G[B])$
\item[3.] From Theorem~\ref{theorem:sat_Tseitin_formulas_to_DNNF} we derive that $|\Sigma^*| = 2^{\Omega(\tw(G))}$ and use $|\Sigma^*| = O(|\Sigma^\ell| \times |\Sigma^r|)$ to conclude.
\end{itemize} 
For convenience we denote $G_A := G[A]$ and $G_B := G[B]$.
In the second step, we can not really control which of $T(G_A)$ or $T(G_B)$ is satisfiable. But the first step frees us from worrying about this: since both $G_A$ and $G_B$ have large treewidth, $\Sigma^*$ have size exponential in the treewidth of $G$ regardless of whether it represents $T(G_A)$ or $T(G_B)$.

The following lemmas will be proved in the next sections.

\begin{lemma}\label{lemma:good_partition}
Let $G = (V,E)$ be a 2-connected graph with maximum degree $\Delta$. There is a partition $(A,B)$ of~$V$ such that $G_A$ is connected, $G_B$ is 2-connected, and $\min(\tw(G_A)),\tw(G_B)) \geq\lfloor\frac{\alpha\tw(G)}{\Delta^2}\rfloor$ where $\alpha > 0$ is a fixed universal constant.
\end{lemma}

\begin{lemma}\label{lemma:conditioning_for_Tseitin_formulas}
Let $T(G,c)$ be a Tseitin formula with $G$ connected and a partition $(A,B)$ of $V$ such that both $G_A$ and $G_B$ are connected. Then for every assignment $a$ to $X_{E(A,B)}$ there are $c^a_A : A \rightarrow \{0,1\}$ and $c^a_B : B \rightarrow \{0,1\}$ such that 
$$
T(G,c)|a = T(G_A, c^a_A) \land T(G_B,c^a_B).
$$ 
Moreover, if $T(G,c)$ is unsatisfiable then either $T(G_A, c^a_A)$ or $T(G_B, c^a_B)$ is unsatisfiable, but not both. Which of the two formulas is satisfiable depends on whether the number of variables that $a$ maps to $1$ is odd or even.
\end{lemma}

\begin{lemma}\label{lemma:technical_lemma}
Let \textup{Apply}$(\Sigma^\ell$, $\Sigma^r$, $\land)$ be the last step of a str-DNNF($\land$,$r$) refutation of $T(G)$ where $G$ is 2-connected. Assume that there is a partition $(A,B)$ of $V$ such that $G_A$ is connected, $G_B$ is 2-connected, and both have treewidth at least~2. Then there is a str-DNNF of size $O(|\Sigma^\ell|\times|\Sigma^r|)$ computing a satisfiable Tseitin formula whose graph is $G_A$ or~$G_B$.
\end{lemma}

\begin{proof}[Proof of Theorem~\ref{theorem:main_result}]
First, using Lemmas~\ref{lemma:good_partition} and~\ref{lemma:technical_lemma} and Theorem~\ref{theorem:sat_Tseitin_formulas_to_DNNF}, we prove the result when $G$ is 2-connected. Let $\Delta$ be an upper bound on the maximum degree of all our graphs. Fix a graph $G=(V,E)$ and consider the partition $(A,B)$ of $V$ given by Lemma~\ref{lemma:good_partition}. Let $k = \tw(G)$ and $n = |E(G)|$. We can choose the constant hidden in $2^{\Omega(k)}$ of the statement so that the theorem becomes trivial whenever $\lfloor\alpha k/\Delta^2\rfloor < 2$, so we assume $\lfloor\alpha k/\Delta^2\rfloor \ge 2$ in the remainder. 

The conditions on $(A,B)$ described in Lemma~\ref{lemma:technical_lemma} are met so we obtain a str-DNNF $\Sigma^*$ computing a \emph{satisfiable} Tseitin formula $T(G_U)$ for some $U \in \{A,B\}$ with $|\Sigma^*| \leq \gamma\times|\Sigma^\ell|\times |\Sigma^r|$ for some $\gamma > 0$. Now Theorem~\ref{theorem:sat_Tseitin_formulas_to_DNNF} says that there is a constant $\beta > 0$ such that $|\Sigma^*| \geq 2^{\beta k/\Delta^3}/n$. So we have $\min(|\Sigma^\ell|,|\Sigma^r|) \geq 2^{\beta k/2\Delta^3}/(\gamma n)$. This completes the proof in the case where $G$ is 2-connected.

Now we show how to go from the general case to the case where $G$ is 2-connected. Assume $G$ has a 1-separator $\{u\}$ and let $U_1,\dots,U_s$ be the vertex sets of the connected components of $G$ after removal of $u$. We know from Theorem~\ref{theorem:bodlaender} that there is some $i \in [s]$ such that $\tw(G[U_i \cup \{u\}]) = \tw(G)$, say $i = 1$. Now there is a proper subset $E' \subset E(u)$ such that removing $E'$ from $G$ yields two connected components $G_A$ and $G_B$, with $U_i \cup \{u\} \subseteq A$. So $E' = E(A,B)$ and, by Lemma~\ref{lemma:conditioning_for_Tseitin_formulas}, we can choose an assignment $a$ to $E'$ such that $T(G)|a = T(G_A) \land T(G_B)$ where $T(G_B)$ is satisfiable and $T(G_A)$ unsatisfiable.

Let $a_B$ be a satisfying assignment of $T(G_B)$. Using Lemma~\ref{lemma:conditioning_for_strDNNF_compilation} we can condition any str-DNNF($\land$,$r$) refutation of $T(G)$ on the assignment $a \cup a_B$ to obtain a str-DNNF($\land$,$r$) refutation of $T(G_A)$ without size increase. $G_A$ has fewer 1-separators than $G$ and $\tw(G_A) = \tw(G)$. We repeat the procedure until obtaining a str-DNNF($\land$,$r$) refutation of $T(G')$, with $G'$ a subgraph of $G$ that has the same treewidth of $G$ and has no 1-separator. So $G'$ is 2-connected, and the refutation of $T(G')$ obtained is at most as large as that of $T(G)$ we have started from.
\end{proof}
\section{Graph Bi-Partition with Large Treewidth on Both Sides (Lemma~\ref{lemma:good_partition})}

Lemma~\ref{lemma:good_partition} is shown with the help of Theorem~\ref{theorem:tw_partition} below combined with Theorem~\ref{theorem:bodlaender}. For space reasons we defer the proof to Appendix A.
We here discuss some of the underlying graph theory, in particular the following result.

\begin{theorem}\label{theorem:tw_partition}
There exists a constant $0 < \alpha \leq 1$ such that, for all graphs $G = (V,E)$ with maximum degree at most $\Delta$, there is a partition $(A,B)$ of $V$ such that $\tw(G[A]) \geq \lfloor\frac{\alpha\tw(G)}{\Delta^2}\rfloor$ and $\tw(G[B]) \geq \lfloor\frac{\alpha\tw(G)}{\Delta^2}\rfloor$. 
\end{theorem}

\noindent To illustrate Theorem~\ref{theorem:tw_partition}, we look at the particular case of grid graphs. The $n\times n$ grid has treewidth $n-1$ and maximum degree $4$. It is straightforward to partition its vertices to obtain an $n\times \lfloor n/2 \rfloor$ grid on one side, and an $n\times \lceil n/2 \rceil$ on the other. Using this partition for $(A,B)$ we see that $G[A]$ and $G[B]$ both have an $\lfloor n/2 \rfloor \times \lfloor n/2 \rfloor$ induced grid and therefore both have treewidth at least $\lfloor n/2 \rfloor - 1 \geq (n-1)/4$. Of course the constant $\alpha$ in the theorem is way smaller than $4$.

The proof of Theorem~\ref{theorem:tw_partition} is technical and is deferred to Appendix~C, here we just provide some arguments to justify its veracity. Theorem~\ref{theorem:tw_partition} is an adaptation of the following result of Chekuri of Chuzhoy~(\citeyear{ChekuriC13}). 

\begin{theorem}\label{theorem:Chekuri_Chuzhoy_thm}
Let $h$ and $r$ be integers and let $G = (V,E)$. There are positive constants $\beta$ and $c$ such that, if $h^3r \leq \beta \frac{\tw(G)}{\log^c(\tw(G))}$, then there is an \emph{efficient} algorithm to partition $V$ into $(V_1,\dots, V_h)$, with $\tw(G[V_i]) \geq r$ true for all $i \in [h]$.
\end{theorem}

\noindent Theorem~\ref{theorem:tw_partition} is almost a subcase of Theorem~\ref{theorem:Chekuri_Chuzhoy_thm} with $h = 2$. The only problem is that in Theorem~\ref{theorem:tw_partition}, $r$ would be roughly $\alpha/\Delta^2$ and thus independent of the treewidth, which is not the case in Theorem~\ref{theorem:Chekuri_Chuzhoy_thm} because of the divisor $\log^c(\tw(G))$. A careful examination of Chekuri and Chuzhoy's proof shows that the log-divisor has two reasons: (1) a preprocessing of $G$ to decrease its degree and (2) the use of an approximation algorithm to make their partition efficiently computable. Since we work with graphs of bounded degree and only care about the existence of a partition and not its computation, we can adapt the proof for $h = 2$ and make some other adjustments to get rid of the $\log^c(\tw(G))$ to obtain Theorem~\ref{theorem:tw_partition}, see Appendix~C for details.

\setcounter{lemma}{4}

\section{Graph Partitions for Tseitin formulas and subformulas (Lemma~\ref{lemma:conditioning_for_Tseitin_formulas})}

In this section we prepare for the proof of Lemma~\ref{lemma:technical_lemma} by recalling some results on how Tseitin formulas behave when we disconnect the underlying graph. The variables of a Tseitin formula $T(G,c)$ uniquely identify the edges of its underlying graph $G$. After assigning the variable $x$ corresponding to the edge $e := uv$ in $T(G,c)$ and removing the negated literals and the satisfied clauses, the new formula is a Tseitin formula $T(G',c')$, but this time for the graph $G' = (V,E\setminus\{e\})$. If $x$ is assigned $0$, then the new charge function is the same as the old one, that is, $c' = c$. Otherwise if $x$ is assigned $1$, then $c'$ coincides with $c$ on all vertices except $u$ and $v$, that is, $c' = c + 1_u + 1_v \mod 2$. By induction, conditioning $T(G,c)$ on a partial assignment of its variables yields a new Tseitin formula whose underlying graph is $G$ without the corresponding edges. We focus on variable conditionings that disconnect $G$.
\begin{proof}[Proof of Lemma~\ref{lemma:conditioning_for_Tseitin_formulas}]
Let $G = (V,E)$. $T(G,c)|a$ be the CNF obtained by removing from $T(G,c)$ all clauses containing a literal set to $1$ by $a$, and removing all literals set to $0$ by $a$ from the remaining clauses. $T(G,c)|a$ is exactly $T(G',c^a)$ with $G' = (V,E\setminus E(A,B))$ and $c^a = c + \sum_{x_{uv} : a(x_{uv}) = 1} 1_u + 1_v \mod 2$.

$G_A$ and $G_B$ are the only two connected components of~$G'$ so $T(G',c^a) = T(G_A, c^a_A) \land T(G_B,c^a_B)$ where $c^a_A$ and $c^a_B$ are the restrictions of $c^a$ to $A$ and $B$ respectively. 

Now if $T(G,c)$ is unsatisfiable, then so is $T(G_A, c^a_A) \land T(G_B,c^a_B)$. Let $S_A$ and $S_B$ be the sums of all $c^a(u)$ for $u$ in $A$ and $B$, respectively. By Lemma~\ref{lemma:Tseitin_formula_satisfiability} we have $S_A = 1 \mod 2$ or $S_B = 1 \mod 2$. Observe that $\sum_{u \in V} c^a(u) = S_A + S_B$ and that $\sum_{u \in V} c^a(u) = \sum_{u \in V} c(u) = 1 \mod 2$. So either $S_A = 0 \mod 2$ or $S_B = 0 \mod 2$ holds. Again by Lemma~\ref{lemma:Tseitin_formula_satisfiability}, it follows that either $T(G_A,c^a_A)$ or $T(G_B,c^a_B)$ is satisfiable.

All edges whose variables are assigned values by $a$ have one endpoint in $A$ and the other in $B$. Let $\textup{card}(a) := |\{x_{uv} : a(x_{uv}) = 1\}|$. Then looking at the expression of $c^a$ we see that $S_A =  \sum_{u \in A} c^a(u) = \textup{card}(a) + \sum_{u \in A} c(u)$ and $S_B =  \textup{card}(a) + \sum_{u \in B} c(u)$. Since the parity of $S_A$ and $S_B$ decides the satisfiability of $T(G_A,c^a_A)$ and $T(G_B,c^a_B)$, and since $S_A \neq S_B \mod 2$, we get that the parity of $\textup{card}(a)$ decides which Tseitin formula is satisfiable.
\end{proof}
str-DNNF in the refutation of $T(G,c)$ represent subformulas of $T(G,c)$. For $F$ such a subformula, given a partition $(A,B)$ of $V$ and an assignment $a$ to $X_{E(A,B)}$, the CNF formula $F|a$ is of the form $F^a_A \land F^a_B$ where $F^a_A$ is a subformula of $T(G_A,c^a_A)$ and $F^a_B$ is a subformula of $T(G_B,c^a_B)$.

\setcounter{lemma}{6}
\begin{lemma}\label{lemma:conditioning_for_strDNNF}
Let $\Sigma$ be a str-DNNF representing a subformula $F$ of $T(G)$. Let $(A,B)$ be a partition of $V$ and $a$ be an assignment to $X_{E(A,B)}$ such that $F|a$ is satisfiable, then there are str-DNNF $\Sigma_A$ and $\Sigma_B$ of size at most $|\Sigma|$ and with the same vtree as $\Sigma$, that represent $F^a_A$ and $F^a_B$ respectively.
\begin{proof}
This follows from conditioning being feasible without size increase nor vtree modification on str-DNNF~\cite{PipatsrisawatD08}. Let $T$ be the vtree of $\Sigma$. First we can obtain a str-DNNF $\Sigma'$ equivalent to $\Sigma|a \equiv F|a = F^a_A \land F^a_B$ of size at most $|\Sigma|$ and that respects $T$. Since $F|a$ is satisfiable and since the variables of $F^a_A$ and $F^a_B$ are disjoint, we have an assignment $a'$ to the variables of $F^a_B$ that satisfies $F^a_B$ and such that $(F|a)|a' = F^a_A$. So from $\Sigma'$ we can obtain a str-DNNF equivalent to $\Sigma'|a \equiv F^a_A$, of size at most $|\Sigma'|$, and whose vtree is $T$. The argument works analogously for~$F^a_B$.
\end{proof}
\end{lemma}

%
%

\section{From Unsatisfiable to Satisfiable Tseitin formulas (Lemma~\ref{lemma:technical_lemma})}

We call a constraint $\chi$ \emph{incomplete} in a CNF formula $F$ when $\clause(\chi) \cap \clause(F) \neq \clause(\chi)$. Clearly a subformula of $T(G)$ has incomplete constraints if and only if it is a proper subformula of $T(G)$. All str-DNNF in a str-DNNF($\land$,$r$) compilation of $T(G)$, except the last one, have incomplete constraints.


The proof of Lemma~\ref{lemma:technical_lemma} intuitively works by considering two cases as follows. In the first case, we assume that one of $\Sigma^r$ and $\Sigma^\ell$ contains the constraints for $B$ almost completely, say this is the case for $\Sigma^\ell$. We choose an assignment $a$ to $X_{E(A,B)}$ such that the resulting Tseitin formula $T(G_B, c^a_B)$ is satisfiable. Then we can extract from $\Sigma^\ell$ a str-DNNF and conjoin to this str-DNNF the few missing constraints with increasing is size too much, so that it computes $T(G_B, c^a_B)$.

In the second case, several constraints for $B$ are incomplete in both $\Sigma^\ell$ and $\Sigma^r$. In that case, we can choose an assignment $a$ to $X_{E(A,B)}$ such that the subformulas made of constraints for $A$ used in the construction of $\Sigma^r$ are satisfiable, 
and the same is true for $\Sigma^\ell$. Then we can conjoin suitably processed versions of $\Sigma^r$ and $\Sigma^\ell$ to get a str-DNNF representation of $T(G_A, c^a_A)$ without increasing the size too much. More formally, we consider the following two cases:
\begin{itemize}
\item[1.] For \emph{some} $\Sigma \in \{\Sigma^\ell,\Sigma^r\}$, at most two constraints for vertices in $B$ are incomplete in $\Sigma$.
\item[2.] For \emph{every} $\Sigma \in \{\Sigma^\ell,\Sigma^r\}$, at least three constraints for vertices in $B$ are incomplete in $\Sigma$.
\end{itemize}

\setcounter{lemma}{5}
\begin{lemma}[\textbf{Case 1}]
Use the notation of Lemma~\ref{lemma:technical_lemma}. If for some $\Sigma \in \{\Sigma^\ell, \Sigma^r\}$ at most two constraints of $T(G)$ for vertices of $B$ are incomplete in $\Sigma$, then there is a str-DNNF of size $O(|\Sigma|)$ computing a satisfiable Tseitin formula whose underlying graph is $G_B$.
\end{lemma}
\begin{proof}
$\Sigma$ is satisfiable, so there is an assignment $a$ to $X_{E(A,B)}$ such that $\Sigma|a$ is satisfiable. Let $F$ be the CNF whose clauses are used to construct $\Sigma$ in the refutation. With Lemma~\ref{lemma:conditioning_for_strDNNF} we obtain a str-DNNF $\Sigma_B$ equivalent to $F^a_B$ and such that $|\Sigma_B| = O(|\Sigma|)$.
Let $B'$ be the set of vertices in $B$ whose constraints are incomplete in $\Sigma$. By assumption $|B'| \leq 2$. For all $v \in B \setminus B'$, the constraint of $T(G_B,c^a_B)$ for $v$ is complete in $\Sigma_B$ because the constraint of $T(G)$ for $v$ is complete in $\Sigma$. If $B' = \emptyset$ then all constraints of $T(G_B,c^a_B)$ are complete in $\Sigma_B$, so $\Sigma_B \equiv T(G_B,c^a_B)$ and we are done. 

Now assume that $B'$ contains two vertices: $B' = \{v,w\}$ and let $\chi'_{v}$, $\chi'_{w}$ be the constraints of $T(G_B,c^a_B)$ for $v$ and $w$ (for the case of one vertex, just take $v = w$). All constraints but $\chi'_{v}$ and $\chi'_{w}$ are complete in $\Sigma_B$, so 
$$
\Sigma_B \land \chi'_{v} \land \chi'_{w} \equiv T(G_B,c^a_B)
$$
Let $T$ be the vtree for $\Sigma_B$. Lemma~\ref{lemma:parity_constraint_to_strDNNF} gives us str-DNNF $D_{v}$ and $D_{w}$ computing $\chi'_{v}$ and $\chi'_{w}$ respectively, of size $O(\Delta)$, and both respecting $T$. We get a str-DNNF $D$ structured by $T$ and equivalent to $T(G_B,c^a_B)$ of size $O(|\Sigma_B|\times \Delta^2) = O(|\Sigma|)$ by conjoining $D_{v}$ and $D_{w}$ to~$\Sigma_B$.
\end{proof}

\setcounter{lemma}{5}
\begin{lemma}[\textbf{Case 2}]
Use the notation of Lemma~\ref{lemma:technical_lemma}. If for every $\Sigma \in \{\Sigma^\ell, \Sigma^r\}$ at least three constraints of $T(G)$ for vertices of $B$ are incomplete in $\Sigma$, then there is a str-DNNF of size $O(|\Sigma^\ell|\times|\Sigma^r|)$ computing a satisfiable Tseitin formula whose underlying graph is $G_A$.
\begin{proof}

Let $F^\ell$ and $F^r$ be the CNF formulas whose clauses were used to construct $\Sigma^\ell$ and $\Sigma^r$, respectively. Apply($\Sigma^\ell$, $\Sigma^r$, $\land$) is the last apply of the refutation so there must be $F^\ell \land F^r = T(G)$. Our aim is to find an assignment $a$ to $X_{E(A,B)}$ such that $T(G_A,c^a_A)$, $\Sigma^\ell|a$, and $\Sigma^r|a$ are satisfiable. If such an assignment exists, then using Lemma~\ref{lemma:conditioning_for_strDNNF} we could obtain str-DNNF $\Sigma^\ell_A$ and $\Sigma^r_A$ that represent $(F^\ell)^a_A$ and $(F^r)^a_A$ respectively. Since $\Sigma^\ell$ and $\Sigma^r$ have the same vtree, so would $\Sigma^\ell_A$ and $\Sigma^r_A$. So we could construct a str-DNNF computing $\Sigma^\ell_A \land \Sigma^r_A \equiv (F^\ell)^a_A \land (F^r)^a_A = T(G_A,c^a_A)$ in time $|\Sigma^\ell_A| \times |\Sigma^r_A| \leq |\Sigma^\ell| \times |\Sigma^r|$, thus finishing the proof. 

It remains to find this assignment $a$. Take $\Sigma \in \{\Sigma^\ell, \Sigma^r\}$ and let $F \in \{F^\ell, F^r\}$ be the corresponding CNF. By Lemma~\ref{lemma:conditioning_for_Tseitin_formulas}, if $T(G_A,c^a_A)$ is satisfiable then $T(G_B,c^a_B)$ is unsatisfiable. Then we have $\Sigma|a \equiv F^a_A \land F^a_B$ where $F^a_A$ is satisfiable since it is a subformula of $T(G_A,c^a_A)$, and we want $F^a_B$ to be satisfiable as well. The following claims help us find $a$ such that $F^a_B$ is satisfiable. The proofs are deferred to Appendix~B. Note that Claim~\ref{claim:if_2-connected_then_MUS} is a folklore result on Tseitin formulas.

\begin{claim}\label{claim:if_2-connected_then_MUS}
Since $G_B$ is a 2-connected graph, the \emph{proper} subformulas of \emph{any} Tseitin formula $T(G_B)$ are \emph{all} satisfiable.
\end{claim}

\begin{claim}\label{claim:incomplete_constraints_incomplete_clause_set}
Let $F$ be a proper subformula of $T(G)$. Take $C_v \in \clause(\chi_v)$ not in $F$ and denote by $C'_v$ its restriction to $X_{E(A,B)}$. If both $G_A$ and $G_B$ have treewidth at least 2, then for every assignment $a$ to $X_{E(A,B)}$ that falsifies $C'_v$, the constraint $\chi_v|a$ is incomplete in $F|a$.
%
\end{claim}

If we can find $a$ such that $T(G_A,c^a_A)$ is satisfiable and such that $F^a_B$ is a \emph{proper} subformula of $T(G_B,c^a_B)$, i.e., not all constraints of $T(G_B,c^a_B)$ are complete in $F^a_B$. Then by the above claims, $F^a_B$ will be satisfiable. Then $\Sigma|a \equiv F^a_A \land F^a_B$ will be satisfiable as well since $\var(F^a_A) \cap \var(F^a_B) = \emptyset$. Recall that this must hold for both  $F = F^\ell$ and $F = F^r$.

By assumption there is $u^r \in B$ whose constraint is incomplete in $\Sigma^r$ and there are $u^\ell, v^\ell, w^\ell \in B$ whose constraints are incomplete in $\Sigma^\ell$. The latter three vertices are distinct, so at least two of them are different from $u^r$. Suppose, without loss of generality, that $u^r \neq v^\ell$ and $u^r \neq w^\ell$. For convenience, rename $u = u^r$, $v = v^\ell$ and $w = w^\ell$. 

Let $C_u$ be a clause of $\chi_u$ missing from $\clause(\Sigma^r)$ and let $C_v$ and $C_w$ be clauses of $\chi_v$ and $\chi_w$ missing from $\clause(\Sigma^\ell)$. We denote $C_u = C'_u \lor C''_u$ where $C'_u$ is the restriction of $C_u$ to $\var(X_{E(A,B)})$. Note that $C'_u$ may be empty. Define $C_v = C'_v \lor C''_v$ and $C_w = C'_w \lor C''_w$ similarly. Let $E'(u)$, $E'(v)$ and $E'(w)$ be the set of edges corresponding to $\var(C'_u)$, $\var(C'_v)$ and $\var(C'_w)$, respectively. By definition, all three sets are subsets of $E(A,B)$. 
\begin{claim} We have $E(A,B) \neq E'(u) \cup E'(v)$ or $E(A,B) \neq E'(u) \cup E'(w)$.
\begin{proof}
If $E'(u) = \emptyset$ or $E'(v) = \emptyset$ or $E'(w) = \emptyset$, then the claim holds because otherwise $E(A,B)$ would be a subset or $E(u)$, or a subset of $E(v)$, or a subset of $E(w)$, which is not possible since $G$ is 2-connected.

Otherwise, if neither $E'(u)$ nor $E'(v)$ nor $E'(w)$ is empty, then the three sets are pairwise disjoint since $u,v,w \in B$. So if $E(A,B) = E'(u) \cup E'(v)$ were to hold, then we would have $E(A,B) \neq E'(u) \cup E'(w)$ because otherwise $E'(v) = E'(w) \neq \emptyset$ would hold, which is impossible.
\end{proof}
\end{claim}
Suppose, w.l.o.g., that $E(A,B) \neq E'(u) \cup E'(v)$. Let $a'_u$ and $a'_v$ be the assignments to $\var(C'_u)$ and $\var(C'_v)$ that falsify $C'_u$ and $C'_v$, respectively (if $C'_u$ is empty, then so is $a'_u$). Conditioning $T(G)$ on $a'_u \cup a'_v$ gives an unsatisfiable Tseitin formula on the graph obtained by removing $E'(u)\cup E'(v)$ from $G$. Call that graph $G'$. Since $G$, $G_A$, and $G_B$ are connected, and since $E'(u) \cup E'(v)$ is a \emph{proper} subset of $E(A,B)$, we have that $G'$ is connected. So by Lemma~\ref{lemma:conditioning_for_Tseitin_formulas}, we have an assignment $a$ to $X_{E(A,B)}$ that extends $a'_u \cup a'_v$ and such that $T(G)|a = T(G_{A},c^a_{A}) \land T(G_B,c^a_B)$ where $T(G_{A},c^a_{A})$ is satisfiable and $T(G_B,c^a_B)$ is unsatisfiable.

Remember that $C_u$ and $C_v$ are missing from $\clause(\Sigma^r)$ and $\clause(\Sigma^\ell)$, respectively, and that $C'_u$ and $C'_v$ are their restrictions to $X_{E(A,B)}$. By construction,~$a$ falsifies $C'_u$ and $C'_v$, therefore by Claim~\ref{claim:incomplete_constraints_incomplete_clause_set} the constraint for the vertices $u$ and $v$ are incomplete in $\Sigma^r|a$ and $\Sigma^\ell|a$, respectively. It follows, since $u$ and $v$ belong to $B$, that $(F^\ell)^a_B$ and $(F^r)^a_B$ are proper subformulas of $T(G_B,c^a_B)$. Then since $G_B$ is 2-connected, Claim~\ref{claim:if_2-connected_then_MUS} entails that both $(F^\ell)^a_B$ and $(F^r)^a_B$ are satisfiable, and therefore $\Sigma^\ell|a$ and $\Sigma^r|a$ are satisfiable.
\end{proof}
\end{lemma}

\section{Conclusion}
In the past, experimental works hinted at the inefficiency of the bottom-up approach for compiling some inputs into specific languages like OBDD or SDD. In this paper, we provide theoretical arguments that support the idea that the inefficiency of bottom-up compilation resides in the bottom-up paradigm itself. We propose a framework for compilation that targets the very general language of str-DNNF, puts no constraint on the order in which clauses are conjoined, and allows on-the-fly restructuring of the str-DNNF. Despite these degrees of freedom, we have found a class of CNF formulas that have constant-size str-DNNF representations and proved that they require exponential time and space to be compiled with the bottom-up approach.

In the future, it would be interesting to better understand how the size of intermediate results in bottom-up compilation is impacted by the order in which clauses are conjoined. For example, can it be shown theoretically when the heuristics from~\cite{NarodytskaW07} perform well? Can similar heuristics also be used in the construction of SDD?

\section{Acknowledgments} This work has been partly supported by the PING/ACK project of the French National Agency for Research (ANR-18-CE40-0011).

\bibliography{main.bib}

\appendix

\section{Appendix A: Proof of Lemma 4}

\setcounter{lemma}{3}
\begin{lemma}
Let $G = (V,E)$ be a 2-connected graph with maximum degree $\Delta$. There is a partition $(A,B)$ of~$V$ such that $G[A]$ is connected, $G[B]$ is 2-connected, and $\min(\tw(G[A])),\tw(G[B])) \geq\lfloor\frac{\alpha\tw(G)}{\Delta^2}\rfloor$ with $\alpha > 0$ fixed.
\begin{proof}[Proof of Lemma~\ref{lemma:good_partition}]
First take $(A_0,B_0)$ is in Theorem~\ref{theorem:tw_partition}. $G$ being connected, we can assume that both $A_0$ and $B_0$ are connected. First set $(A,B) = (A_0,B_0)$. If $B$ is 2-connected then we are done. Otherwise if $B$ is not 2-connected then there is $b \in B$ such that $G[B \setminus\{b\}]$ has several connected components $B_1$, $B_2$, $\dots$, $B_k$. By Theorem~\ref{theorem:bodlaender}, for some connected component, say for $B_1$, there is $\tw(G[B_1 \cup \{b\}]) = \tw(G[B])$. So replace $B$ by $B_1 \cup \{b\}$ and put $B_2$, $\dots$, $B_k$ in $A$. Observe that doing that, the number of 1-separators of $B$ decreases while the treewidth of $G[B]$ is unchanged. Then we can repeat this until $B$ has no 1-separator and is therefore 2-connected. 

It remains to prove that $G[A]$ is connected. Assume that there is a connected component $U \subseteq A$ of $G[A]$ that does not contain $A_0$. By construction, there is $b \in B$ such that removing $b$ from $G[B \cup U]$ disconnects $U$ (from $G[B \cup U]$). Since $U$ is also a connected component in $G[A]$ removing $b$ from $G$ also disconnects $U$ (from $G$). But then $G$ is not 2-connected. A contradiction.
\end{proof}
\end{lemma}

\section{Appendix B: Missing Proofs of the Claims}

The claims' statements are rephrased with more general notations than that used in the proof of Lemma~\ref{lemma:technical_lemma} where they appear.

\setcounter{claim}{0}
\begin{claim}
If $G$ is a 2-connected graph, then the \emph{proper} subformulas of \emph{any} Tseitin formula $T(G,c)$ are \emph{all} satisfiable.
\begin{proof}
This is trivial when $T(G,c)$ is satisfiable, so we assume otherwise. Let $F$ be a proper subformula of $T(G,c)$ and let $C \in \clause(\chi_u)$ be absent from $F$. Let $a$ be the assignment to $\var(C)$ that falsifies $C$. Clearly $a$ falsifies $\chi_u$, and, since $\var(C) = \var(\chi_u)$ holds for any clause of $\chi_u$, $a$ satisfies $\overline{\chi}_u$ but also all clauses of $\chi_u$ distinct from $C$. Consider the Tseitin formula $T(G,c')$ where $\chi_u$ has been replaced by $\overline{\chi}_u$, so where $c' = c + 1_u \mod 2$. By Lemma~\ref{lemma:Tseitin_formula_satisfiability}, $T(G,c')$ is satisfiable. It holds that $(T(G,c)\setminus C)|a = T(G,c')|a = T(G-u,c'')$ where $c''$ is defined on $V \setminus \{u\}$ as $c'' = c' + \sum_{x_{uv} : a(x_{uv}) = 1} 1_v \mod 2$. In other words, $F|a$ is a subformula of $T(G-u,c'')$. Every variable appears in exactly two constraints, so we have $\sum_{v \in V \setminus \{u\}} c''(v) = \sum_{v \in V} c'(v) \mod 2$ and, since 2-connectivity guarantees that $G-u$ is connected, Lemma~\ref{lemma:Tseitin_formula_satisfiability} tells us that $T(G-u,c'')$, and thus $F$, is satisfiable.
\end{proof}
\end{claim}

\begin{claim}
Let $F$ be a proper subformula of $T(G)$. Assume there is a partition $(A,B)$ of $V$ such that $G_A$ and $G_B$ are connected and have treewidth at least~$2$. Take $C_v \in \clause(\chi_v)$ not in $F$ and denote by $C'_v$ its restriction to $X_{E(A,B)}$. If an assignment $a$ to $X_{E(A,B)}$ falsifies $C'_v$ then $\chi_v|a$ is incomplete in $F|a$.
\begin{proof}
Write $C_v = C'_v \lor C''_v$. Note that $C'_v$ may be empty. We will show that $C''_v$ is not in $F|a$. By way of contradiction, assume that this were false, i.e., that $C''_v$ appears in $F|a$. If $C''_v$ has at least two literals, then these are literals for variables in $X_{E(v)}$. Since $C''_v$ is in $F|a$, there is a clause $C = R \lor C''_v$ in $\clause(F) \cap \clause(\chi_v)$ such that $a$ falsifies $R$. But since all clauses of $\chi_v$ are in all variables of $\var(\chi_v)$, we have $\var(R) = \var(C'_v)$, and since $a$ already falsifies $C'_v$ we get that $R = C'_v$. So $C = C_v$, a contradiction.

Now assume $C''_v$ contains a single literal for the variable $x_e$, with $e = uv$. If $C''_v$ is in $F|a$, then there is a clause $C_u = C'_u \lor C''_v$ in $\clause(F) \cap \clause(\chi_u)$ such that $a$ falsifies $C'_u$. Now for every $w \in V$, let $E'(w) := E(w) \cap E(A,B)$ and observe that $E'(u) = E(u) \setminus \{uv\}$ and $E'(v) = E(v) \setminus \{uv\}$. But then $G'=(\{u,v\},\{uv\})$ is a connected component of $G$ after removal of $E(A,B)$. In other words, $G'$ must be $G_A$ or $G_B$. But $G'$ has treewidth $1$, which is too small, another contradiction.
\end{proof}
\end{claim}
\setcounter{claim}{3}

\section{Appendix C: Proof of Theorem~\ref{theorem:tw_partition}}
In this appendix, we prove Theorem~\ref{theorem:tw_partition} on graph bi-partitions with large treewidth on both sides.

\setcounter{theorem}{3}
\begin{theorem}
There exists a constant $0 < \alpha \leq 1$ such that, for all graphs $G = (V,E)$ with maximum degree at most $\Delta$, there is a partition $(A,B)$ of $V$ such that $\tw(G[A]) \geq \lfloor\frac{\alpha\tw(G)}{\Delta^2}\rfloor$ and $\tw(G[B]) \geq \lfloor\frac{\alpha\tw(G)}{\Delta^2}\rfloor$. 
\end{theorem}

\noindent The closest to Theorem~\ref{theorem:tw_partition} we could find in the literature is Theorem~\ref{theorem:Chekuri_Chuzhoy_thm} by Chekuri and Chuzhoy~(\citeyear{ChekuriC13}). Indeed, setting $h = 2$, $V_1 = A$ and $V_2 = B$ in Theorem~\ref{theorem:Chekuri_Chuzhoy_thm} gives a result similar to our except that it holds for graphs of unbounded-degree and that the treewidths of $G[A]$ and $G[B]$ are not proportional to $k = \tw(G)$ but to $k/\text{poly}\log(k)$. To prove Theorem~\ref{theorem:tw_partition} we stick to Chekuri and Chuzhoy's techniques for proving Theorem~\ref{theorem:Chekuri_Chuzhoy_thm}. Looking at their proof, one can see that the $\text{poly}\log(k)$ factor comes from a reduction from unbounded-degree to bounded-degree graphs, and from their efforts to obtain an \emph{efficient} algorithm to compute the partition of $V(G)$. Since we already look at graphs of bounded-degree and since we do not care for efficiency, we can eliminate the $\text{poly}\log(k)$ factor. 

The proof of Theorem~\ref{theorem:tw_partition} is split in subsections due to its intricacy. The first two subsections introduce preliminary notions, then in the next subsection we give two lemmas that help proving Theorem~\ref{theorem:tw_partition}, and the remaining subsections contain the proofs of these two lemmas. 

We define the following numerical constants: $\gamma = 1/2000$, $\beta = 6/\gamma = 12000$ and $\alpha = 1/(200\beta)$ the constant from Theorem~\ref{theorem:tw_partition}. When the set of vertices, the set of edges, or the maximum degree of a graph $G$ is not specified, we denote it by $V(G)$, $E(G)$, and $\Delta(G)$, respectively. We assume $\frac{\alpha \tw(G)}{\Delta^2} \geq 1$, otherwise the theorem is trivial.

\subsection{Well-linkedness}

In a graph $G$, a set $S \subseteq V(G)$ is called \emph{well-linked} when, for every pair $X, Y \subseteq S$ such that $|X| = |Y|$, there exists $|X|$ vertex-disjoint paths from $X$ to $Y$. Note that $X$ and $Y$ are not necessarily distinct and that paths of size zero are allowed. The well-linked number of $G$, denoted $wl(G)$, is the size of the largest well-linked set in $G$. It is known that $\tw(G) \leq wl(G) + 1 \leq 3\times\tw(G)$~\cite{HarveyW17}. 

\setcounter{lemma}{7}
\begin{lemma}\label{lemma:connection_tw-wl}
Let $G$ be a simple graph. There is a set $S^* \subseteq V(G)$ of size $\tw(G) \leq |S^*|+1 \leq 3\times\tw(G)$ such that, for every partition $(A,B)$ of $V(G)$, it holds that 
$$
|E(A,B)| \geq \min(|A \cap S^*|,|B \cap S^*|).
$$
\end{lemma}
\begin{proof}
Let $S^*$ be the largest well-linked set in $G$ and consider a partition $(A,B)$ of $V(G)$. The bounds on $|S^*|$ stem from the relation between $wl(G)$ and $\tw(G)$. Let $S^*_A = S^* \cap A$ and $S^*_B = S^* \cap B$. Assume, without loss of generality, that $|S^*_A| \leq |S^*_B|$. Take an arbitrary subset $Z \subseteq S^*_B$ of size $|S^*_B| - |S^*_A|$ and let $X = S^*_A \cup Z$ and $Y = S^*_B$. There is $X, Y \subseteq S^*$ and $|X| = |Y|$. 

By definition of $S^*$ there are at least $|X| = |S^*_A| + |Z|$ vertex-disjoint paths from $X$ to $Y$, $|S^*_A|$ of which start from $S^* \cap A$ and end in $S^* \cap B$. Since $A \cap B = \emptyset$, these paths all have size at least 1, and since they are vertex-disjoint, there are at least $|S^*_A| = |S^* \cap A| = \min (|S^* \cap A|, |S^* \cap B|)$ edges going from $A$ to $B$. Thus $|E(A,B)| \geq |S^*_A| = |S^* \cap A| = \min (|S^* \cap A|, |S^* \cap B|)$.
\end{proof}

\noindent In the remainder of that appendix, $S^*$ is the subset of $V$ described by Lemma~\ref{lemma:connection_tw-wl}, $k := |S^*|$ and $r := \frac{2\alpha k}{\Delta^2}$. Observe that $r \geq \frac{\alpha \tw(G)}{\Delta^2} \geq 1$.

\begin{lemma}[\citealt{ChekuriC13}]\label{lemma:lower_bound_treewidth}
Let $G$ be a simple graph whose maximum degree is $\Delta$, assume that there is $S \subseteq V(G)$ and $\gamma \in [0,1]$ such that, for every partition $(A,B)$ of $V(G)$, it holds that $|E(A,B)| \geq \gamma \min(|S \cap A|,|S\cap B|)$, then $\tw(G) \geq \frac{\gamma |S|}{3\Delta} - 1$. 
\end{lemma}

\subsection{Acceptable partitions}

In a graph $G$, given a subset $S \subseteq V(G)$, we denote by $out_G(S)$ the set of edges of $G$ that have one endpoint in $S$ and the other in $V(G)\setminus S$. We drop the $G$ subscript when it is clear from context which graph we are working on. Let $\calC = (V_1, \dots, V_s)$ be a partition of $V(G)$ ($s$ is arbitrary). We denote by $G_\calC$ the multigraph with $s$ vertices $\{\nu_1,\dots, \nu_s\}$ such that for each $i \neq j$ (no self-loop) there are $|E(V_i,V_j)|$ edges between $\nu_i$ and $\nu_j$. One can construct $G_\calC$ from $G$ by contracting every subset of vertices $V_i$ into a single vertex, which is $\nu_i$. We then call $G_\calC$ a \emph{contracted multigraph}. We abuse the notations and consider that the vertices of $G_\calC$ are directly $V_1,\dots,V_s$ (so $\calC$ can be seen as $V(G_\calC)$). It will be important that all multigraphs $G_\calC$ considered for different partitions $\calC$ have maximum degree bounded by $\Delta' := \beta r \Delta^2$ where $\beta = 12000$, and many edges. We will show that these properties are guaranteed when the partition $\calC$ is \emph{acceptable}.

\begin{definition}
A partition $\calC = (V_1, \dots, V_s)$ of $V(G)$ is called \emph{acceptable} when $|out(V_i)| \leq \Delta'$ and $|V_i \cap S^*| \leq \frac{k}{2}$ hold for all $i \in [s]$.
\end{definition}

Acceptable partitions of $V(G)$ exist, the simplest one is $\calC = \{\{v\} \mid v \in V(G) \}$. Indeed there is $|\{v\} \cap S^*| \leq 1 \leq \frac{k}{2}$ since $k \geq 2$ and there is $|out(\{v\})| = \deg(v) \leq \Delta \leq \beta r\Delta^2 = \Delta'$ since $\beta, r \geq 1$.

\begin{claim}
Let $\calC$ be an acceptable partition of $V(G)$, then $\Delta(G_\calC) \leq \Delta'$.
\end{claim}
\begin{proof}
Let $\calC = \{V_1, \dots, V_s\}$ and let $\nu_i$ be the vertex of $G_\calC$ corresponding to $V_i$, then $\deg(\nu_i) = |out(V_i)| \leq \Delta'$. 
\end{proof}

\begin{claim}
Let $\calC$ be an acceptable partition of $V(G)$, then $|E(G_\calC)| \geq \frac{k}{4}$.
\end{claim}
\begin{proof}
Let $\calC = \{V_1, \dots, V_s\}$. Assume $|V_1 \cap S^*| \geq |V_2 \cap S^*| \geq \dots \geq |V_s \cap S^*|$ holds. There is $\sum_{i = 1}^s |V_i \cap S^*| = |S^*| = k$. Since $|V_1 \cap S^*| \leq \frac{k}{2}$ we can find $l \in [s-1]$ the largest integer such that $\sum_{i = 1}^l |V_i \cap S^*| \leq \frac{3k}{4}$. We clearly have $\sum_{i = l+1}^s |V_i \cap S^*| \geq \frac{k}{4}$. We also have $\sum_{i = 1}^l |V_i \cap S^*| \geq \frac{k}{4}$ for otherwise there would be $|V_{l+1} \cap S^*| = \sum_{i = 1}^{l+1} |V_i \cap S^*| - \sum_{i = 1}^l |V_i \cap S^*| > \frac{3k}{4} - \frac{k}{4} = \frac{k}{2}$, which can not be.

Let $A = V_1 \cup \dots \cup V_l$ and $B = V_{l+1} \cup \dots \cup V_{s}$. By construction we have $|A \cap S^*| \geq \frac{k}{4}$ and $|B\cap S^*| \geq \frac{k}{4}$. And by definition of $S^*$ there is $|E(A,B)| \geq \min( |A \cap S^*|, |B \cap S^*|) \geq \frac{k}{4}$. When contracting the $V_i$ to obtain $G_\calC$, the edges $E(A,B)$ survive, so $|E(G_\calC)| \geq |E(A,B)| \geq \frac{k}{4}$.
\end{proof}

\subsection{Proof of Theorem~\ref{theorem:tw_partition}}

The proof of Theorem~\ref{theorem:tw_partition} boils down to the following two lemmas, which we will show are correct in later sections.

\begin{lemma}\label{lemma:nice_vertex_partition}
Let $\calC$ be an acceptable partition of $V(G)$, then there is a partition $(\calU_1,\calU_2,\calU_3)$ of $\calC$ such that, for all $i \in \{1,2,3\}$
$$
|E(G_\calC[\calU_i])| \geq \frac{|E(G_\calC)|}{180}.
$$ 
\end{lemma}

Recall that $\calC = (V_1,\dots,V_s)$, where $V_1 \cup \dots \cup V_s = V(G)$. From a partition $(\calU_1,\calU_2,\calU_3)$ of $\calC$ we define a corresponding partition $(U_1,U_2,U_3)$ of $V(G)$ where $U_j$ is obtained by uncontracting all nodes in $\calU_j$. More formally, $U_j  = \bigcup_{V_i \in \calU_j} V_i$. Note that $\calU_j$ is partition of $U_j$.

\begin{lemma}\label{lemma:if_not_wl_then_better_partition}
Let $\calC = (V_1,\dots,V_s)$ be an acceptable partition of $V(G)$. Let $\calU \subseteq \calC$ such that $|E(G_\calC[\calU])| \geq \frac{|E(G_\calC)|}{180}$ and let $U = \bigcup_{V_i \in \calU} V_i$. If $\tw(G[U]) < \lfloor \frac{\alpha \tw(G)}{\Delta^2} \rfloor$ and if $|U \cap S^*| \leq \frac{k}{2}$, then there is a partition $\calU'$ of $U$ such that $\calC' = (\calC \setminus \calU) \cup \calU'$ is an acceptable partition of $V(G)$ and such that $|E(G_{\calC'})| < |E(G_\calC)|$.
\end{lemma}

\begin{proof}[Proof of Theorem~\ref{theorem:tw_partition}]
We know acceptable partitions of $V(G)$ exist. Let $\calC$ be the acceptable partition of $V(G)$ such that $|E(G_\calC)|$ is minimal, that is, for all other acceptable partitions $\calC'$ there is $|E(G_{\calC'})| \geq |E(G_\calC)|$. Let $(\calU_1,\calU_2,\calU_3)$ be the partition of $\calC$ given by  Lemma~\ref{lemma:nice_vertex_partition} and $(U_1,U_2,U_3)$ be the corresponding partition of $V(G)$. 

Assume, without loss of generality, that $\tw(G[U_1]) \geq \tw(G[U_2]) \geq tw(G[U_3])$. If $\tw(G[U_2]) \geq \lfloor \frac{\alpha \tw(G)}{\Delta^2} \rfloor$ then we take $A = U_1$ and $B = U_2 \cup U_3$ and we are done. Suppose otherwise that $\tw(G[U_3]) \leq \tw(G[U_2]) < \lfloor \frac{\alpha \tw(G)}{\Delta^2} \rfloor$. There must be $|U_j \cap S^*| \leq \frac{|S^*|}{2} = \frac{k}{2}$ for some $j \in \{2,3\}$, say for $j = 3$. But then Lemma~\ref{lemma:if_not_wl_then_better_partition} gives a partition $\calU'_3$ of $U_3$ such that $\calC' = \calU_1 \cup \calU_2 \cup \calU'_3$ is an acceptable partition of $V(G)$ such that $|E(G_{\calC'})| < |E(G_\calC)|$,  a contradiction.
\end{proof}

\subsection{Proof of Lemma~\ref{lemma:nice_vertex_partition}}

Lemma~\ref{lemma:nice_vertex_partition} is a consequence of the more general Lemma~\ref{lemma:tool_lemma} below. Lemma~\ref{lemma:tool_lemma} essentially says that if the number edges of a multigraph is greater than a factor of its degree, then there is a partition of its vertices into three parts such that many edges remain in every part.

\begin{lemma}\label{lemma:tool_lemma}
Let $H$ be a multigraph with no self-loop and with $|E(H)| \geq 25\Delta(H)$. There is a partition $V(H) = (V_r,V_b,V_g)$ (red,blue,green) such that $|E(H[V_c])| \geq |E(H)|/180$ holds for all $c \in \{r,b,g\}$.
\end{lemma}

\begin{proof}[Proof of Lemma~\ref{lemma:nice_vertex_partition}]
$\calC$ is an acceptable partition, so $|E(G_\calC)| \geq k/4$ and $\Delta(G_\calC) \leq \Delta' = \beta r \Delta^2 = 2\beta \alpha k \leq k/100 \leq |E(G_\calC)|/25$. Now Lemma~\ref{lemma:nice_vertex_partition} is a direct application of Lemma~\ref{lemma:tool_lemma} with $H = G_\calC$.
\end{proof}

\noindent The proof of Lemma~\ref{lemma:tool_lemma} is probabilistic. It leans on \textbf{Paley-Zygmund inequality} which, given a non-negative random variable $Z$ with finite variance and $\theta \in (0,1]$, is
$$
\Prb{Z \geq \theta \Ex{Z}} \geq (1-\theta)^2\frac{\Ex{Z}^2}{\Ex{Z^2}}
$$

\begin{proof}[Proof of Lemma~\ref{lemma:tool_lemma}]
Let $m = |E(H)|$, $D = \Delta(H)$ and $\eta = \frac{1}{25}$. By assumption $D \leq \eta m $.
The vertices of $H$ are assigned a color in $\{r(ed),b(lue),g(reen)\}$ uniformly at random. An edge $e = uv$ is red when both its endpoints are red, it is blue when both its endpoints are blue, it is green when both its endpoints are green, and otherwise it has no color. Let $X^c_e$ be the event that the edge $e$ has color $c \in \{r,b,g\}$ and let $E_c = \sum_{e \in E(H)} X^c_e$ be the number of edges colored with $c$ after random coloring of the vertices. It is clear that $\Prb{X^c_e} = \frac{1}{9}$ and that $\Ex{E_c} = \frac{m}{9}$. Proving the statement of the lemma means proving that
$$
\Prb{E_r < \frac{m}{180} \textup{ or } E_b < \frac{m}{180} \textup{ or } E_g < \frac{m}{180}} < 1.
$$ 
By union bound, it is sufficient to show that $\Prb{E_c < \frac{m}{180}} < \frac{1}{3}$ holds, for $c$ fixed in $\{r,b,g\}$. We will use Paley-Zygmund inequality to prove $\Prb{E_c \geq \frac{m}{180}} > \frac{2}{3}$, so we need to compute $\Ex{E_c^2}$.
$$
\begin{aligned}
\Ex{E_c^2} 
&= \Ex{\left(\sum\nolimits_{e} X^c_e\right)^2} = \sum_{e \in E(G)} \sum_{e' \in E(G)} \Ex{X^c_eX^c_{e'}}
\\
&= \sum_{e \in E(G)} \sum_{e' \in E(G)}\Prb{X^c_e \textup{ and } X^c_{e'}}
\end{aligned}
$$
Let us look at $\Prb{X^c_e \textup{ and } X^c_{e'}}$ for $e = uv$. 
\begin{itemize}
\item[$\bullet$] If $e'$ has the same endpoints as $e$, so when $e' \in E(u) \cap E(v)$, then the probability is $\Prb{X^c_e \textup{ and } X^c_{e'}} = \frac{1}{9}$. 
\item[$\bullet$] If $e'$ shares exactly one endpoint with $e$, so when $e' \in (E(u) \cup E(v)) \setminus (E(u) \cap E(v))$, then the probability is $\Prb{X^c_e \textup{ and } X^c_{e'}} = \frac{1}{27}$.
\item[$\bullet$] If $e'$ has no endpoint in common with $e$, so when $e' \in E(H) \setminus (E(u) \cup E(v))$, then the probability is $\Prb{X^c_e \textup{ and } X^c_{e'}} = \frac{1}{81}$.
\end{itemize}
So we obtain
$$
\begin{aligned}
\Ex{E_c^2} &=
\sum_{e = uv}\left(\frac{|E(u) \cap E(v)|}{9} \right.
\\
&\quad+\frac{|E(u)| + |E(v)| - 2|E(u) \cap E(v)|}{27}
\\
&\quad\left.+ \frac{m - (|E(u)| + |E(v)| - |E(u) \cap E(v)|)}{81}\right) 
\\
\Ex{E_c^2} &=\frac{m^2}{81} + \frac{4}{81}\sum_{e = uv}|E(u) \cap E(v)| \\
&\quad+\frac{2}{81}\sum_{e = uv}(|E(u)| + |E(v)|)
\end{aligned}
$$
Recall that we are dealing with multigraphs, so $\sum_{e = uv}|E(u) \cap E(v)|$ is not necessarily $m$. But we do have that $|E(u) \cap E(v)|$, $|E(u)|$ and $|E(v)|$ are all at most $D \leq \eta m$ so 
$$
\Ex{E_c^2} \leq \frac{m^2}{81} +\frac{8mD}{81} \leq \frac{m^2}{81} + \frac{8\eta m^2}{81}.
$$
Finally we apply Paley-Zygmund inequality:
$$
\begin{aligned}
\Prb{E_c \geq m/180} &\geq \left(1-\frac{1}{20}\right)^2\frac{\Ex{E_c}^2}{\Ex{E_c^2}} 
\\
&\geq \left(1-\frac{1}{20}\right)^2\frac{1}{1+ 8\eta} > \frac{2}{3}
\end{aligned}
$$
\end{proof}

\subsection{Proof of Lemma~\ref{lemma:if_not_wl_then_better_partition}}

\setcounter{lemma}{10}
\begin{lemma}
Let $\calC = (V_1,\dots,V_s)$ be an acceptable partition of $V(G)$. Let $\calU \subseteq \calC$ such that $|E(G_\calC[\calU])| \geq \frac{|E(G_\calC)|}{180}$ and let $U = \bigcup_{V_i \in \calU} V_i$. If $\tw(G[U]) < \lfloor \frac{\alpha \tw(G)}{\Delta^2} \rfloor$ and if $|U \cap S^*| \leq \frac{k}{2}$, then there is a partition $\calU'$ of $U$ such that $\calC' = (\calC \setminus \calU) \cup \calU'$ is an acceptable partition of $V(G)$ and such that $|E(G_{\calC'})| < |E(G_\calC)|$.
\end{lemma}
\setcounter{lemma}{12}

We are going to construct the new partition $\calU'$ of $U$. Recall that $\frac{6\Delta^ 2 r}{\Delta'} = \frac{6}{\beta} = \gamma < 1$. From Lemma~\ref{lemma:lower_bound_treewidth} we deduce that for all $S \subseteq U$ with $|out(S)| \geq \Delta'$, there is a partition $(A_U,B_U)$ of $U$, such that $|E(A_U,B_U)| < \gamma \min(|A_U \cap S|, |B_U \cap S|)$, for otherwise there would be 
$$
\begin{aligned}
\tw(G[U]) &\geq \frac{\gamma |S|}{3\Delta} - 1 \geq \frac{\gamma |out(S)|}{3\Delta^2} - 1 \geq \frac{\gamma \Delta'}{3\Delta^2} - 1
\\
&= 2r-1 \geq r \geq \frac{\alpha\tw(G)}{\Delta^2}
\end{aligned}
$$
where we have used that $|S| \geq |out(S)|/\Delta$ and $r = 2\alpha k/\Delta^2 \geq \alpha\tw(G)/\Delta^2 \geq 1$.

\paragraph{The Split Function.} We define a routine $Split(Y)$ whose inputs is a subset $Y \subseteq U$ with $|out(Y)| \geq \Delta'$. $Split(Y)$ first chooses the smallest subset $S \subseteq Y$ whose vertices are endpoints of edges in $out(Y)$, and such that $|out(S) \cap out(Y)| \geq \Delta'$. Clearly $|S| \geq \Delta'/\Delta > 1$. Then $Split(Y)$ returns a partition $(A_Y,B_Y)$ of $Y$ such that $|E(A_Y,B_Y)| < \gamma \min(|S \cap A_Y|,|S \cap B_Y|)$. We know such a partition exists because otherwise $\tw(G[U]) \geq \tw(G[Y]) \geq r$ would hold. We always assume that $|out(A_Y)| \leq |out(B_Y)|$. Observe that neither $S \cap A_Y$ nor $S \cap B_Y$ is empty.
%
%
%
%
%
%
%
%
%
%
%
%
%
%

\begin{lemma}\label{lemma:split_lemma_1}
Let $(A_Y,B_Y) = Split(Y)$, then 
\begin{equation}\label{eq:eq1}
|E(A_Y,B_Y)| < \gamma \Delta'
\end{equation}
and 
\begin{equation}\label{eq:eq2}
\begin{aligned}
|E(A_Y,B_Y)| < \gamma \min(&|out(Y) \cap out(A_Y)|,
\\
&|out(Y) \cap out(B_Y)|)
\end{aligned}
\end{equation}
\begin{proof}
Consider the subset $S \subseteq Y$ chosen by $Split(Y)$. Every $v \in S$ is the endpoint of an edge in $out(Y)$ by definition, so $|out(S \cap A_Y) \cap out(Y)| \geq |S \cap A_Y|$ and $|out(S \cap B_Y) \cap out(Y)| \geq |S \cap B_Y|$. Thus $|out(A_Y) \cap out(Y)| \geq |S \cap A_Y|$ and $|out(B_Y) \cap out(Y)| \geq |S \cap B_Y|$ hold. $Split(Y)$ returns the partition $(A_Y,B_Y)$ such that $|E(A_Y,B_Y)| < \gamma \min(|S \cap A_Y|,|S \cap B_Y|)$. Combining this inequality with the ones we have just obtained gives~(\ref{eq:eq2}).

There is $|out(S \cap A_Y) \cap out(Y)| < \Delta'$ and $|out(S \cap B_Y) \cap out(Y)| < \Delta'$, for otherwise $|S|$ would not be minimal. Thus $|E(A_Y,B_Y)| < \gamma \min(|S \cap A_Y|,|S \cap B_Y|) < \gamma \min(|out(S \cap A_Y) \cap out(Y)|,|out(S \cap B_Y) \cap out(Y)|) < \gamma\Delta'$.
\end{proof}
\end{lemma}

\begin{lemma}\label{lemma:split_lemma_2} Let $(A_Y,B_Y) = Split(Y)$, then $|out(A_Y)| \leq \frac{|out(Y)|}{2(1-\gamma)}$ and $|out(B_Y)| \leq |out(Y)|$.
\begin{proof}
By definition of $Split$ there is $|out(A_Y)|\leq |out(B_Y)|$ and $|E(A_Y,B_Y)| < \gamma |out(A_Y) \cap out(Y)|$.

For the first part, observe that $|out(A_Y)| + |out(B_Y)| - 2|E(A_Y,B_Y)| = |out(Y)|$ so 
\[
\begin{aligned}
|out(Y)| &\geq 2|out(A_Y)| - 2|E(A_Y,B_Y)| \\
		 &\geq 2|out(A_Y)| - 2\gamma|out(A_Y) \cap out(Y)| \\
		 &\geq 2(1-\gamma)|out(A_Y)|
\end{aligned}
\]
For the second part, observe that $|out(B_Y)| = |out(B_Y) \cap out(Y)| + |E(A_Y,B_Y)|$ so 
\[
\begin{aligned}
|out(B_Y)| &\leq |out(B_Y) \cap out(Y)| + \gamma |out(A_Y) \cap out(Y)|\\
		   &\leq \gamma|out(Y)| + (1-\gamma)|out(B_Y) \cap out(Y)| \\
		   &\leq \gamma|out(Y)| + (1-\gamma)|out(B_Y)| \\
\end{aligned}
\]
and therefore $\gamma |out(B_Y)| \leq \gamma|out(Y)|$ holds.
\end{proof}
\end{lemma}

\begin{algorithm}[tb]
\caption{BetterPartition($U$)}
\begin{algorithmic}[1] 
\STATE Let $P = (U)$ 
\WHILE {\textup{there exists $Y \in P$ such that $|out(Y)| \geq \Delta'$}}
\STATE Let $(A_Y,B_Y) = Split(Y)$ 
\STATE Remove $Y$ from $P$ and add $A_Y$ and $B_Y$ to $P$
\ENDWHILE
\STATE $\calU' = \emptyset$
\FOR {$Y \in P$}
\STATE Add all connected components of $G[Y]$ to $\calU'$
\ENDFOR
\RETURN $\calU'$
\end{algorithmic}
\end{algorithm}

\paragraph{A Partition Algorithm.} Recall that we have a partition $\calU$ of $U \subseteq V(G)$ and that we want to replace it with an new partition $\calU'$ such that $\calC' = (\calC \setminus \calU) \cup \calU'$ is an acceptable partition of $V(G)$ with $|E(G_{\calC'})| < |E(G_\calC)|$. The new partition $\calU'$ is given by the algorithm BetterPartition($U$). The algorithm starts from the partition $P = (U)$ of size 1. Then, as long as $P$ has a component $Y$ with a border too big (i.e., $|out(Y)| \geq \Delta'$), the algorithm calls $Split$ to divide $Y$ in two parts and replaces $Y$ by the two parts. The algorithm ends because $Split(Y)$ returns a partition $(A_Y,B_Y)$ where neither $A_Y$ nor $B_Y$ is empty, so $Y$ is replaced by two smaller sets. In the worst case we would reach the point where every the sets in $P$ contains a single vertex, i.e., is of the form $\{v\}$, and the \emph{while} loop ends since $|out(\{v\})| = \deg(v) \leq \Delta < \Delta'$. The trace of all splits occurring in BetterPartition($U$) forms a rooted binary tree $T$ where each internal $t \in T$ corresponds to a subset of $U$. We encode this with a mapping $\lambda$ from nodes of $T$ to subset of $U$: if $\lambda(t) = Y$ and $(A_Y,B_Y) = Split(Y)$, then $\lambda(t_l) = A_Y$ and $\lambda(t_r) = B_Y$ where $t_l$ and $t_r$ are the children of $t$. See for instance Figure~\ref{fig:splitting_sequence}~(a), it represents a sequence of splits whose trace is the tree shown Figure~\ref{fig:splitting_sequence}~(b).

\begin{lemma}\label{lemma:new_acceptable_partition}
Let $\calU'$ be the output of BetterPartition($U$). Then $\calC' = (\calC \setminus \calU) \cup \calU'$ is an acceptable partition
\begin{proof}
Consider $P$ at the end of the \emph{while} loop in BetterPartition($U$). Let $Y \in P$ and let $Y^* \subseteq Y$ be such that $G[Y^*]$ is a connected component of $G[Y]$. It is readily verified that $|out(Y^*)| \leq |out(Y)| < \Delta'$. By assumption there is $|U \cap S^*| \leq \frac{k}{2}$ so $|Y^* \cap S^*| \leq \frac{k}{2}$ holds as well. Now $\calU'$ is a partition of $U$ whose elements are sets like $Y^*$. By what we have shown, since $\calC$ is an acceptable partition, so is $\calC' = (\calC \setminus \calU) \cup \calU'$. 
\end{proof}
\end{lemma}
Lemma~\ref{lemma:new_acceptable_partition} is a first step towards proving Lemma~\ref{lemma:if_not_wl_then_better_partition}. It remains to show that $|E(G_{\calC'})| < |E(G_\calC)|$. First observe that $|E(G_{\calC'})| =$
\begin{equation*}
\begin{aligned}
&|E(G_\calC)| - |E(G_\calC[\calU])| - |out_{G_\calC}(\calU)| + \big|\bigcup_{Y \in \calU'} out_G(Y)\big|
\\
&= |E(G_\calC)| - |E(G_\calC[\calU])| - |out_G(U)| + \big|\bigcup_{Y \in P} out_G(Y)\big|
\end{aligned}
\end{equation*}

$\bigcup_{Y \in P} out(Y)$ contains all the edges of $out(U)$ plus the edges $E(A_Y,B_Y)$ for every split $(A_Y,B_Y) = Split(Y)$ done by the algorithm. So $|\bigcup_{Y \in P} out(Y)| $ equals $|out(U)|$ plus some value $M$.
\begin{equation}\label{eq:eq3}
|E(G_{\calC'})| = |E(G_\calC)| - |E(G_\calC[\calU])| + M
\end{equation}
\begin{figure*}[t]
\begin{subfigure}{0.8\textwidth}
\centering
\begin{subfigure}{0.22\textwidth}
\centering
\begin{tikzpicture}[scale=1.2]
\draw (0,0) circle (1.05);
\node[circle,fill=white,inner sep=3] (u) at (-0.68,0.8) {}; 
\node[circle,fill,inner sep=1,label={above:\small{$u$}}] (u) at (-0.68,0.6) {}; 
\node[circle,fill,inner sep=1,label={above:$v$}] (v) at (-1.15,0.9) {}; 
\draw (u) -- (v);
\node[font=\small] (Y) at (0,0) {$Y$};
\end{tikzpicture}
\end{subfigure}
\begin{subfigure}{0.22\textwidth}
\centering
\begin{tikzpicture}[scale=1.2]
\draw[rounded corners=5pt] (0, 0.05) -- (1, 0.05)  arc[start angle=0, end angle=180, radius=1] -- (0, 0.05);
\draw[rounded corners=5pt] (0, -0.05) -- (1, -0.05)  arc[start angle=0, end angle=-180, radius=1] -- (0, -0.05);

\node[circle,fill=white,inner sep=3] (u) at (-0.68,0.8) {}; 
\node[circle,fill,inner sep=1,label={above:\small{$u$}}] (u) at (-0.68,0.6) {}; 
\node[circle,fill,inner sep=1,label={above:$v$}] (v) at (-1.15,0.9) {}; 
\draw (u) -- (v);

\node[font=\small] (A1) at (0,0.33) {$A_1$};
\node[font=\small] (B1) at (0,-0.33) {$B_1$};

\end{tikzpicture}
\end{subfigure}
\begin{subfigure}{0.22\textwidth}
\centering
\begin{tikzpicture}[scale=1.2]
\draw[rounded corners=5pt] (0.05,0.8) -- (0.05, 0.05) -- (1, 0.05) arc[start angle=0, end angle=90, radius=0.95] -- (0.05, 0.8);
\draw[rounded corners=5pt] (-0.05,0.8) -- (-0.05, 0.05) -- (-1, 0.05) arc[start angle=180, end angle=90, radius=0.95] -- (-0.05, 0.8);
\draw[rounded corners=5pt] (0, -0.05) -- (1, -0.05)  arc[start angle=0, end angle=-180, radius=1] -- (0, -0.05);

\node[circle,fill=white,inner sep=3] (u) at (-0.68,0.8) {}; 
\node[circle,fill,inner sep=1,label={above:\small{$u$}}] (u) at (-0.68,0.6) {}; 
\node[circle,fill,inner sep=1,label={above:$v$}] (v) at (-1.15,0.9) {}; 
\draw (u) -- (v);

\node[font=\small] (B2) at (-0.4,0.33) {$B_2$};
\node[font=\small] (A2) at (0.4,0.33) {$A_2$};

\end{tikzpicture}
\end{subfigure}
\begin{subfigure}{0.22\textwidth}
\centering
\begin{tikzpicture}[scale=1.2]
\draw[rounded corners=5pt] (0.05,0.8) -- (0.05, 0.05) -- (1, 0.05) arc[start angle=0, end angle=90, radius=0.95] -- (0.05, 0.8);

\draw[rounded corners=5pt] (-0.5,0.7) -- (-0.5, 0.05) -- (-1, 0.05) arc[start angle=180, end angle=120, radius=0.95] -- (-0.5, 0.7);

\draw[rounded corners=4pt] (-0.05, 0.06) -- (-0.05, 1) -- (-0.4, 0.94) -- (-0.4, 0.06) -- cycle;

\draw[rounded corners=5pt] (0, -0.05) -- (1, -0.05)  arc[start angle=0, end angle=-180, radius=1] -- (0, -0.05);

\node[circle,fill=white,inner sep=3] (u) at (-0.68,0.8) {}; 
\node[circle,fill,inner sep=1,label={above:\small{$u$}}] (u) at (-0.68,0.6) {}; 
\node[circle,fill,inner sep=1,label={above:$v$}] (v) at (-1.15,0.9) {}; 
\draw (u) -- (v);

\node[font=\small] (A3) at (-0.7,0.33) {$A_3$};
\node[font=\small] (B3) at (-0.22,0.33) {$B_3$};

\end{tikzpicture}
\end{subfigure}
\caption{\small{$(A_1,B_1) = Split(Y)$, $(A_2,B_2) = Split(A_1)$, $(A_3,B_3) = Split(B_2)$.}}\label{fig:splitting_sequence_a}
\end{subfigure}
\begin{subfigure}{0.19\textwidth}
\begin{tikzpicture}
\def\x{0.7};
\def\y{0.7};
\node[font=\small] (Y) at (0,0) {$Y$};
\node[font=\small] (A1) at (-\x,-\y) {$A_1$};
\node[font=\small] (B1) at (\x,-\y) {$B_1$};
\node[font=\small] (A2) at (-2*\x,-2*\y) {$A_2$};
\node[font=\small] (B2) at (0,-2*\y) {$B_2$};
\node[font=\small] (A3) at (-\x,-3*\y) {$A_3$};
\node[font=\small] (B3) at (\x,-3*\y) {$B_3$};

\draw (Y) -- (A1);
\draw (Y) -- (B1);
\draw (A1) -- (A2);
\draw (A1) -- (B2);
\draw (B2) -- (A3);
\draw (B2) -- (B3);
\end{tikzpicture}
\caption{\small{The trace of the splits.}}\label{fig:splitting_sequence_b}
\end{subfigure}
\caption{\small{A sequence of splits and its trace.}}\label{fig:splitting_sequence}
\end{figure*}

\subsubsection{Charging Scheme.} We bound $M$ by replaying BetterPartition($U$) with a charging scheme that puts non-negative real numbers, called \emph{charges}\footnote{these charges are \emph{not} related to the charge functions of Tseitin formulas}, on the edges of $E(G[U]) \cup out(U)$. Initially all charges are null. When a split $(A_Y,B_Y) = Split(Y)$ occurs, each edge in $|E(A_Y,B_Y)|$ adds a charge $|out(A_Y) \cap out(Y)|^{-1}$ to every edge in $out(A_Y) \cap out(Y)$ (which we recall is not empty). So for the split $(A_Y,B_Y) = Split(Y)$, a total charge of $|E(A_Y,B_Y)|$ is created in the graph. This is the only way to add charges in the graph therefore, when the algorithm ends, the total charge equals $M$. 

Existing charges are also moved in the graph in such a way that when the algorithm ends, only edges of $out(U)$ have non-zero charges. This will allow us to bound $M$ as a fraction of $|out(U)|$. Movements of charges occur during splits along with charges creation.  When the split $(A_Y,B_Y) = Split(Y)$ occurs, instead of having every edge in $E(A_Y,B_Y)$ give a charge $|out(A_Y) \cap out(Y)|^{-1}$, we decide that every edge $e' \in E(A_Y,B_Y)$ that already have a charge $c_{e'}$ adds a charge $(1+c_{e'})/|out(A_Y) \cap out(Y)|$ to every edge in $out(A_Y) \cap out(Y)$. So a $e'$ contributes a total charge $1 + c_{e'}$ to $out(A_Y) \cap out(Y)$. The charge of $e'$ is then reset to $0$ since its old charge $c_{e'}$ has been stored in $out(A_Y) \cap out(Y)$. Algorithm Charging($T, \lambda$) shows an implementation of the charging scheme. Its input $(T,\lambda)$ encodes the splits done during the course of BetterPartition($U$).

\begin{algorithm}[tb]
\caption{Charging($T, \lambda$)}
\begin{algorithmic}[1] 
\STATE All $c_e$ are set to $0$
\STATE Mark all leaves of $T$ as visited
\WHILE {\textup{there is $t \in T$ not marked such that $t_r$ and $t_l$ are marked}}
\STATE Let $Y = \lambda(t)$, $A_Y = \lambda(t_l)$, $B_Y = \lambda(t_r)$ 
\FOR {$e \in out(A_Y) \cap out(Y)$}
\STATE Set $c_e = c_e + \sum_{e' \in E(A_Y,B_Y)} \frac{1+c_{e'}}{|out(A_Y) \cap out(Y)|}$
\ENDFOR
\FOR {$e' \in E(A_Y,B_Y)$}
\STATE Set $c_{e'} = 0$
\ENDFOR
\STATE Mark $t$ as visited
\ENDWHILE
\end{algorithmic}
\end{algorithm}

Fix an edge $e = uv$. Suppose that some set $Y$ is split into $(A_Y,B_Y)$ and that $e$ is in $out(Y)$. Either $u \in Y$ and $v \not\in Y$ in which case charges are added to $e$ if and only if $u \in A_Y$, or $v \in Y$ and $u \not\in Y$ in which case charges are added to $e$ if and only if $v \in A_Y$. In the first case we say that $e$ is charged via $u$, in the other case $e$ is charged via $v$. See for instance Figure~\ref{fig:splitting_sequence}~(a). In the leftmost figure, $Y$ is represented by a circle and we have an edge $e = uv \in out(Y)$ with $u \in Y$. Next $Y$ is split into $(A_1,B_1) = Split(Y)$. Since $u \in A_1$, a charge is added to $e$ via $u$. Next $A_1$ is split into $(A_2,B_2) = Split(A_1)$, but no charge is added to $e$ because $u$ is in $B_2$, not in $A_2$. 

\begin{lemma}\label{lemma:charging_M}
Let $\calU'$ be the output partition of BetterPartition($U$) and let $(T,\lambda)$ be the trace of the splits done by BetterPartition($U$). Let $\calC' = (\calC \setminus \calU) \cup \calU'$. After running Charging($T,\lambda$), the total charge in $G$ equals $M := |E(G_{\calC'})| - |E(G_{\calC})| + |E(G_{\calC}[\calU])|$
\begin{proof}
$M$ equals the sum of $|E(A_Y,B_Y)|$ over all splits $(A_Y,B_Y) = Split(Y)$ done during the course of BetterPartition($U$). 

Fix an iteration of the \emph{while} loop of Charging($T, \lambda$). Let $(A_Y,B_Y) = Split(Y)$ be the split for that iteration. In the first \emph{for} loop, a charge of $|E(A_Y,B_Y)| + \sum_{e' \in E(A_Y,B_Y)} c_{e'}$ is created and added to the graph. Then the second \emph{for} loop sets the charges for all $e' \in E(A_Y,B_Y)$ to $0$, so a charge $\sum_{e' \in E(A_Y,B_Y)} c_{e'}$ is lost. Thus the total charge created during that iteration of the \emph{while} loop is $|E(A_Y,B_Y)|$. There is one iteration per split of done by BetterPartition($U$) so when Charging($T, \lambda$) finishes the total charge in the graph equals~$M$.
\end{proof}
\end{lemma}

\begin{lemma}\label{lemma:charging_rec}
Let $(T,\lambda)$ be the trace of the splits done by BetterPartition($U$). Let $e = uv$ be an edge with at least one endpoint in $U$. After running Charging($T,\lambda$) the total charge on $e$ is $c_e \leq 9\gamma$. Moreover if $e \not\in out(U)$ then $c_e = $.
\begin{proof}

If $e \not\in out(U)$ and there is no $Y$ in the trace of BetterPartition$(U)$ such that $e \in E(A_Y,B_Y)$, then $c_e$ never moves from its initial value, that is $0$. If $e \not\in out(U)$ and $e \in E(A_Y,B_Y)$ for some $(A_Y,B_Y) = Split(Y)$, then $c_e$ is set to $0$ at line $9$ of Charging($T$,$\lambda$) and it stays $0$ because for all $Y'$ that are processed after $Y$, either $Y$ is disjoint from $Y'$, or $Y \subseteq A_{Y'}$, or $Y \subseteq B_{Y'}$, and therefore $e$ can never be in $out(A_{Y'}) \cap out(Y')$.

This leaves us the case $e \in out(U)$. We show that $c_e \leq 9\gamma$ by proving the following more general result.

\begin{claim} At every moment during the course of Charging($T,\lambda$), there exists an $r \in \mathbb{N}$ such that for each $e$ there is $c_e \leq \sum_{j = 1}^r (8\gamma)^j$.

\begin{proof}
The proof is by induction. The claim is clearly true at the beginning of the algorithm when all $c_e$ are $0$. For the general case let $e = uv$. We define $c_e(u)$ and $c_e(v)$ as the charges given to $e$ via $u$ and via $v$. Clearly $c_e = c_e(u) + c_e(v)$. We focus on $c_e(u)$. Consider the sequence $Y_{z+1} \subset Y_{z} \subset \dots \subset Y_1$ where $|out(Y_{z+1})| < \Delta'$, $u \in Y_{z+1}$, $v \not\in Y_1$ (so $e \in out(Y_i)$ for all $i$) and such that, for all $i$, $e$ is charged via $u$ when $Y_i$ is processed in the \emph{while} loop. Note that $Y_{i+1}$ may not be $A_{Y_i}$, see for instance Figure~\ref{fig:splitting_sequence}, it shows that $u \in A_3 \subset B_2 \subset A_1$ so there is an $i$ such that $Y_i = A_1$ and $Y_{i+1} = A_3$ is different from $A_{Y_i} = A_2$. Despite this observation, Lemma\ref{lemma:split_lemma_2} still gives us that $|out(Y_{i+1})| \leq \frac{|out(Y_i)|}{2(1-\gamma)}$ for all $i \in [z]$. Since we also have $|out(Y_i)| \geq \Delta'$ for all $i \in [z]$ (otherwise $Y_i$ would not be split) we obtain 
$$|out(Y_i)| \geq (2(1-\gamma))^{z-i}|out(Y_z)| \geq (2(1-\gamma))^{z-i} \Delta'$$

\begin{claim}\label{claim:charge_for_creating_Yi} For all  $i \leq z$, it holds that
$$\frac{|E(Y_i, Z_i)|}{|out(Y_i) \cap out(Y_i \cup Z_i)|} \leq \frac{\gamma}{1-\gamma}\left(\frac{1}{2(1-\gamma)}\right)^{z-i}
$$
\end{claim}
\begin{proof}
$|out(Y_i) \cap out(Y_i \cup Z_i)| = |out(Y_i)| - |E(Y_i,Z_i)| \geq |out(Y_i)| - \gamma|out(Y_i) \cap out(Y_i \cup Z_i)| \geq (1-\gamma)|out(Y_i)| \geq (1-\gamma)\Delta' (2(1-\gamma))^{z-i}$. We then use Lemma~\ref{lemma:split_lemma_1}~(\ref{eq:eq1}) to conclude.
\end{proof}

When processing the split of $Y_i$ the charge added to $e$ via $u$ is 
\[
\begin{aligned}
\frac{\sum_{e' \in E(Y_i,Z_i)}(1+c_{e'})}{|out(Y_i) \cap out(Y_i \cup Z_i)|} & \leq \frac{(1+\max_{e'} c_{e'})|E(Y_i,Z_i)|}{|out(Y_i) \cap out(Y_i \cup Z_i)|}
\end{aligned}
\]
By induction hypothesis, $1+\max_{e'} c_{e'} \leq \sum_{j = 0}^r (8\gamma)^j$ and we know by Claim~\ref{claim:charge_for_creating_Yi} that $|E(Y_i,Z_i)|/|out(Y_i) \cap out(Y_i \cup Z_i)| \leq \frac{\gamma}{1-\gamma}\left(\frac{1}{2(1-\gamma)}\right)^{z-i}$, so when creating $Y_i$ the charge added to $e$ via $u$ is at most 
$$
\frac{\gamma}{1-\gamma}\bigg(\sum_{j = 0}^r (8\gamma)^j\bigg)\left(\frac{1}{2(1-\gamma)}\right)^{z-i} 
$$
We miss a bound on the charge added to $e$ when creating $Y_{z+1}$. We just use that it is fewer than $\gamma$ by Lemma~\ref{lemma:split_lemma_2}~(\ref{eq:eq2}). So we obtain that the charge added to $e = uv$ via $u$ when creating $Y_1,Y_2, \dots Y_{z+1}$ is at most
\[
\begin{aligned}
c_e(u) &\leq \gamma + \frac{\gamma}{1-\gamma}\bigg(\sum_{j = 0}^r (8\gamma)^j\bigg) \sum_{i = 1}^{z} \left(\frac{1}{2(1-\gamma)}\right)^{z-i} \\
&\leq \gamma + \frac{2\gamma}{1-2\gamma}\sum_{j = 0}^r (8\gamma)^j 
\leq \gamma + 3\gamma\sum_{j = 0}^r (8\gamma)^j
\end{aligned}
\]
Where $2\gamma/(1-2\gamma) \leq 3\gamma$ comes from $\gamma = 1/2000$. The bound holds for $c_e(v)$ as well so $c_e = c_e(u) + c_e(v)$ is at most
$$
\begin{aligned}
2\gamma + 6\gamma\sum_{j = 0}^r (8\gamma)^j = 8\gamma + 6\gamma\sum_{j = 1}^r (8\gamma)^j 
\\
\leq 8\gamma + 8\gamma\sum_{j = 1}^r (8\gamma)^j = \sum_{j = 1}^{r+1} (8\gamma)^j
\end{aligned}
$$
This finishes proving the claim 
\end{proof}
\end{claim}

So for $e \in out(U)$ we have that $c_e \leq \sum_{j = 1}^\infty (8\gamma)^j = \frac{1}{1-8\gamma} - 1$, which is fewer than $9\gamma$ when $\gamma = 1/2000$. 
\end{proof}
\end{lemma}

Now we can finish the proof of Lemma~\ref{lemma:if_not_wl_then_better_partition}.

\begin{proof}[Proof of Lemma~\ref{lemma:if_not_wl_then_better_partition}]
Run BetterPartition($U$) and let $\calU'$ be the output partition of $U$. Let $\calC' = (\calC \setminus \calU) \cup \calU'$. By Lemma~\ref{lemma:new_acceptable_partition}, $\calC'$ is an acceptable partition. It remains to show $|E(G_{\calC'})| < |E(G_\calC)|$. Recall Equation~(\ref{eq:eq3}):
$$
|E(G_{\calC'})| = |E(G_\calC)| - |E(G_\calC[\calU])| + M
$$
Let $(T,\lambda)$ be the trace of the splits done by BetterPartition($U$). By Lemma~\ref{lemma:charging_M}, the total charge in the graph after running Charging($T,\lambda$) is $M$ which, by Lemma~\ref{lemma:charging_rec} is at most $9\gamma|out(U)|$. Using this bound in Equation~(\ref{eq:eq3}) yields
$$
\begin{aligned}
|E(G_{\calC'})| 
&\leq |E(G_\calC)| - |E(G_\calC[\calU])| + 9\gamma|out(U)|
\\
&\leq |E(G_\calC)| - |E(G_\calC[\calU])| + 9\gamma|E(G_\calC)|
\\
&\leq |E(G_\calC)| + \left(9\gamma - \frac{1}{180}\right)|E(G_\calC)| 
\\
&< |E(G_\calC)|
\end{aligned}
$$
\end{proof} 
\end{document}